\documentclass[11pt]{article}
\usepackage[utf8]{inputenc}
\usepackage[english]{babel}
\usepackage[T1]{fontenc}
\usepackage[margin=1.15in]{geometry}
\usepackage{amsmath}
\usepackage{amsthm}
\usepackage{amsfonts}
\usepackage{graphicx}
\usepackage{float}
\usepackage{pgfplots}
\pgfplotsset{compat=1.14, set layers}
\usepackage{bm}
\usepackage{mathtools}
\usepackage{subfig}
\usepackage{enumitem}
\usepackage{setspace}
\usepackage{caption} 
\usepackage{float}
\usepackage{mathtools}
\usepackage{xspace,xcolor}
\usepackage{thmtools}
\usepackage{framed}
\usepackage[bottom]{footmisc}
\usepackage{soul}

\definecolor{refcolor}{rgb}{0.23, 0.27, 0.29}
\usepackage[breaklinks,colorlinks,urlcolor={refcolor},citecolor={refcolor},linkcolor={refcolor}]{hyperref}
\usepackage[capitalize, nameinlink]{cleveref}

\newtheorem{theorem}{Theorem}[section]
\newtheorem*{theorem*}{Theorem}
\newtheorem{lemma}[theorem]{Lemma}

\newtheorem{corollary}[theorem]{Corollary}
\newtheorem{definition}[theorem]{Definition}

\newtheorem{claim}[theorem]{Claim}
\newtheorem{observation}[theorem]{Observation}

\newtheorem{example}[theorem]{Example}
\newtheorem*{fact*}{Fact}

\newtheorem*{lemma*}{Lemma}
\newtheorem*{definition*}{Definition}

\Crefname{paragraph}{Paragraph}{Paragraphs}
\Crefname{observation}{Observation}{Observations}

\newcommand{\mpc}{\textsf{MPC}\xspace}
\newcommand{\local}{\textsf{LOCAL}\xspace}

\newcommand{\lcl}{\textsf{LCL}\xspace}
\newcommand{\lcls}{\textsf{LCL}s\xspace}
\newcommand{\llle}{\textsf{LLL}\xspace}

\newcommand{\dist}{\text{dist}}
\newcommand{\diam}{\text{diam}}
\newcommand{\ad}{\hat{D}}
\newcommand{\arr}{\xrightarrow{}}
\newcommand{\narr}{\not\xrightarrow{}}
\newcommand{\larr}{\xleftarrow{}}

\newcommand{\cc}{\textsf{MAX-ID-Solver}\xspace}

\newcommand{\clst}{{\sf CompressLightSubTrees}\xspace}
\newcommand{\cp}{{\sf CompressPaths}\xspace}
\newcommand{\acp}{{\sf AdvancedCompressPaths}\xspace}

\newcommand{\dclst}{{\sf DecompressLightSubTrees}\xspace}
\newcommand{\dcp}{{\sf DecompressPaths}\xspace}

\newcommand{\pro}{{\sf ProbeDirections}\xspace}

\newcommand{\expo}{{\sf Exp}\xspace}
\newcommand{\fd}{{\sf fullDirs}\xspace}
\newcommand{\ld}{{\sf largestDir}\xspace}

\newcommand{\con}{{\sf Contract}\xspace}
\newcommand{\ex}{{\sf Expand}\xspace}
\newcommand{\rk}{{\sf Rake}\xspace}
\newcommand{\ins}{{\sf Insert}\xspace}

\newcommand{\pre}{{\sf Preprocessing}\xspace}
\newcommand{\post}{{\sf Postprocessing}\xspace}

\newcommand{\act}{\textsf{active}\xspace}
\newcommand{\happy}{\textsf{happy}\xspace}
\newcommand{\unhappy}{\textsf{unhappy}\xspace}
\newcommand{\full}{\textsf{full}\xspace}
\newcommand{\sad}{\textsf{sad}\xspace}

\newcommand{\maxid}{\textsf{MAX-ID}\xspace}
\newcommand{\coco}{\textsf{CC}\xspace}

\newcommand{\sr}{{\sf LCLSolver}\xspace}
\newcommand{\css}{{\sf CountSubtreeSizes}\xspace}
\newcommand{\gs}{{\sf GatherSubtrees}\xspace}
\newcommand{\cs}{{\sf CompressSubtrees}\xspace}
\newcommand{\dcs}{{\sf DecompressSubtrees}\xspace}

\newcommand{\cd}{{\sf CountDistances}\xspace}

\newcommand{\inn}{{\operatorname{in}}}
\newcommand{\out}{{\operatorname{out}}}
\newcommand{\sinn}{\Sigma_{\inn}}
\newcommand{\sout}{\Sigma_{\out}}

\newcommand{\ginn}{g_{\inn}}

\newcommand{\ID}{\underline{\sf ID}\xspace}
\newcommand{\id}{{\sf id}\xspace}

\newcommand{\vs}{1 \textrm{vs.}\ 2 cycles\xspace}
\newcommand{\diamConnectivity}{$D$-diameter $s$-$t$ path-connectivity\xspace}

\DeclareMathOperator*{\Exp}{\mathbb{E}}
\DeclareMathOperator*{\argmax}{arg\,max}

\newcommand{\poly}{\operatorname{\text{{\rm poly}}}}

\newif\ifdraft
\drafttrue

\newcommand{\falgo}[2]{
	\vspace{1mm}
	\begin{framed}
		\noindent #1
		\vspace{-1.5mm}
		
		\noindent \hrulefill
		\begin{enumerate}[leftmargin=*]
			#2
		\end{enumerate}
		\vspace{-2.5mm}
	\end{framed}
}

\begin{document}

\begin{center}
	{\huge \bf Optimal Deterministic Massively Parallel Connectivity on Forests} \\ \vspace{1cm}

\begin{minipage}[H]{14.5cm} 
	{\large \textbf{Alkida Balliu}, Gran Sasso Science Institute -- \href{mailto:alkida.balliu@gssi.it}{\texttt{alkida.balliu@gssi.it}}} \vspace{0.5mm}\\
	{\large \textbf{Rustam Latypov\footnotemark}, Aalto University -- \href{mailto:rustam.latypov@aalto.fi}{\texttt{rustam.latypov@aalto.fi}}} \vspace{0.5mm}\\
	{\large \textbf{Yannic Maus}, TU Graz -- \href{mailto:yannic.maus@ist.tugraz.at}{\texttt{yannic.maus@ist.tugraz.at}}} \vspace{1mm}\\
	{\large \textbf{Dennis Olivetti}, Gran Sasso Science Institute -- \href{mailto:dennis.olivetti@gssi.it}{\texttt{dennis.olivetti@gssi.it}}} \vspace{0.5mm}\\
	{\large \textbf{Jara Uitto}, Aalto University -- \href{mailto:jara.uitto@aalto.fi}{\texttt{jara.uitto@aalto.fi}}} \vspace{0.5mm}\\
\end{minipage}

\vspace{5mm}
\begin{minipage}[H]{13.3cm}
\begin{center}
	{\bf Abstract} \\ 
\end{center}

We show fast deterministic algorithms for fundamental problems on forests in the challenging low-space regime of the well-known Massive Parallel Computation (\mpc) model. A recent breakthrough result by Coy and Czumaj~[STOC'22] shows that, in this setting, it is possible to deterministically identify connected components on graphs in $O(\log D + \log\log n)$ rounds, where $D$ is the diameter of the graph and $n$ the number of nodes. The authors left open a major question: is it possible to get rid of the additive $\log\log n$ factor and deterministically identify connected components in a runtime that is completely independent of $n$?\\

We answer the above question in the affirmative in the case of forests. 
We give an algorithm that identifies connected components in $O(\log D)$ deterministic rounds. 
The total memory required is $O(n+m)$ words, where $m$ is the number of edges in the input graph, which is optimal as it is only enough to store the input graph. We complement our upper bound results by showing that $\Omega(\log D)$ time is necessary even for component-unstable algorithms, conditioned on the widely believed \vs conjecture. 
Our techniques also yield a deterministic forest-rooting algorithm with the same runtime and memory bounds. \\

Furthermore, we consider Locally Checkable Labeling  problems (\lcls), whose solution can be verified by checking the $O(1)$-radius neighborhood of each node. We show that any \lcl problem on forests can be solved in $O(\log D)$ rounds with a canonical deterministic algorithm, improving over the $O(\log n)$ runtime of Brandt, Latypov and Uitto [DISC'21]. We also show that there is no algorithm that solves all \lcl problems on trees asymptotically faster.

\end{minipage}
\end{center}

\vfill
\thispagestyle{empty}
\footnotetext{Supported in part by the Academy of Finland, Grant 334238}

\newpage
\thispagestyle{empty}
\tableofcontents

\newpage
\pagenumbering{arabic}

\section{Introduction}
\label{sec:intro}

Graphs offer a versatile abstraction to relational data and there is a growing demand for processing graphs at scale.
One of the most central graph problems in massive graph processing is the detection of connected components of the input graph.
This problem both captures challenges in the study of the fundamentals of parallel computing and has a variety of practical applications.
In this work, we introduce new parallel techniques for finding connected components of a graph.
Furthermore, we show that our techniques can be applied to solve a broad family of other central graph problems.

The Massively Parallel Computation (\mpc) model~\cite{KarloffSV10} is a mathematical abstraction of modern frameworks of parallel computing such as Hadoop~\cite{White2009}, Spark~\cite{Zaharia2010}, MapReduce~\cite{Dean2008}, and Dryad~\cite{Isard2007}. 
In the \mpc model, we have $M$ machines that communicate in synchronous rounds.
In each round, every machine receives the messages sent in the previous round, performs (arbitrary) local computations, and is allowed to send messages to any other machine.
Initially, an input graph of $n$ nodes and $m$ edges is distributed among the machines.
At the end of the computation, each machine needs to know the output of each node it holds, e.g., the identifier of its connected component.
We work in the low-space regime, where the \emph{local memory} $S$ of each machines is limited to $n^\delta$ words of $O(\log n)$ bits, where $0 < \delta < 1$. A word is enough to store a node or a machine identifier from a polynomial (in $n$) domain.
The local memory restricts the amount of data a machine initially holds and is allowed to send and receive per round.
Furthermore, we focus on the most restricted case of \emph{linear total memory}, i.e., $S \cdot M = \Theta(n + m)$.
Notice that $\Omega(n + m)$ words are required to store the input graph.

In recent years, identifying connected components of a graph has gained a lot of attention. 
As a baseline, the widely believed \vs conjecture states that it takes $\Omega(\log n)$ rounds to tell whether the input graph is a cycle of $n$ nodes or two cycles with $n/2$ nodes~\cite{Roughgarden18, Ghaffari2019, Behnezhad2019}.
We note that proving any unconditional lower bounds seems out of reach as any non-constant lower bound in the low-space \mpc model for any problem in P would imply a separation between $\textrm{NC}^1$ and P~\cite{Roughgarden18}.
It has been shown that this conjecture also implies conditional hardness of detecting connected components in time $o(\log D)$ on the family graphs with diameter at most $D$~\cite{Behnezhad2019, ccderandom}.

This bound has been almost matched in a sequence of works.
First, a randomized $O(\log D \cdot \log \log_{m/n} n)$ time algorithm was designed in~\cite{Andoni2018}.
This was further improved to $O(\log D + \log \log_{m/n} n)$ in~\cite{Behnezhad2019} and derandomized with the same asymptotic runtime in~\cite{ccderandom}.
All of the aforementioned algorithms require only $O(n+m)$ words of global memory.
A fundamental question is whether the runtime \emph{necessarily} depends on $n$ for some range of $m$; we give evidence towards a negative answer.
We show that in the case of forests, we can identify the connected components of a graph in $O(\log D)$ time, which we show to be \emph{optimal} under the \vs conjecture.

\begin{framed} \vspace{-5mm}
    \paragraph{Connected Components on Forests.} 
    Consider the family of forests with component-wise maximum diameter $D$. There is a deterministic low-space \mpc algorithm to find the connected components in time $O(\log D)$. The algorithm uses $O(n + m)$ global memory.
    Under the \vs conjecture, this is optimal.
\end{framed}
 
\paragraph{Sparsification and Dependence on $n$.}
In previous works on connected components, the algorithms have an inherent dependency on the total number of nodes $n$ in the input graph.
There is a technical reason for this dependency, also in the context of problems beyond connected components.
A common algorithm design pattern is to first \emph{sparsify} the input graph, i.e., the graph is made much smaller and the problem is solved in the sparser instance~\cite{Andoni2018, GU19, componentstable, Czumaj2020}.
Then, it is shown that a solution to the original input can be recovered from a solution on the sparsified graph.
As an example, a method to sparsify graphs for connectivity is to perform node/edge contractions, that make the graph smaller and preserve connectivity.

In this pattern, the denser the input graph is, the more global memory the algorithm has on the sparsified graph, relatively speaking.
In the aforementioned previous works, the base $m/n$ of the logarithm can be replaced by $F/n$, where $F$ is the global memory.
Hence, if $F = n^{1 + \Omega(1)}$, the dependency on $n$ disappears.
This suggests that the hardest instances are \emph{sparse} graphs, as the $n$ dependency in the runtime of $O(\log D + \log \log_{F / n} n)$ becomes better the larger the global memory $F$ is.
A limitation to solving connected components through independent node/edge contractions comes from the global memory bound.
If the graph is already sparse, then the sparsification cannot make the graph any sparser, and hence we do not have an advantage in terms of global memory on the sparsified graph.
In the case that $m = O(n)$ and the global memory is linear in $n$, the best we could hope for in the first round of contractions is to drop a constant fraction of the nodes.
The low-level details for the reasons behind this can be extracted from the analysis of~\cite{Andoni2018, GU19, componentstable, Czumaj2020}.
Through the relative increase in global memory,  the (remainder) graph size can be bounded by  $n \cdot 2^{-2^i}$ in the $i$th round of contractions, which leads to an $\Omega(\log \log n)$ runtime.

In previous works, there is even more evidence towards sparse graphs being the hardest instances.
Recently, it was shown that lower bound results from the \local model of distributed message passing can be lifted to \mpc under certain conditions~\cite{Ghaffari2019, componentstable}.
In the \local model, almost all hardness results are obtained on trees or in high-girth graphs, implying lower bounds on forests with potentially many connected components~\cite{KuhnMW16, BBHORS21, BBKOmis, BBOrules, BBKO22, BGKO2022}.
It was shown that a \emph{component-stable} algorithm cannot solve a problem $\pi$ faster than in $O(\log T(n, \Delta))$, where $T(n, \Delta)$ is the complexity of $\pi$ in the \local model\footnote{The \local algorithm is allowed to access shared randomness.} on a graph with $n$ nodes and maximum degree $\Delta$.
Roughly speaking, an \mpc algorithm is component-stable if the output on each node $u$ only depends on the size of the graph and the connected component of $u$ (see \Cref{def:componentStability} for more details~\cite{componentstable}).
While these methods do not yield unconditional hardness in the \mpc model, we face similar difficulties in sparse graphs in the \mpc model as in the message passing models.

\paragraph{Rooted Forests and Applications to Locally Checkable Problems.}
We believe that our technique to obtain connected components is of interest beyond solving the connectivity problem.
For example, through minor adjustments to our technique, we obtain an algorithm that roots an (unrooted) input forest.
Furthermore, we show that in a rooted tree, all \emph{Locally Checkable Labeling (\lcl)} problems can be solved very efficiently through a canonical algorithm.
This generalizes to forests and gives an algorithm that can be executed on each connected component independently of the other components.

 \begin{framed} \vspace{-5mm}
 	\paragraph{Locally Checkable Labelings on Forests.} 
 	On the family of forests with component-wise maximum diameter $D$, all \lcl problems can be solved deterministically in $O(\log D)$ time in the low-space \mpc model with $O(n + m)$ global memory.
 	Under the \vs conjecture, this is optimal.
 \end{framed}

A range of central graph problems, in particular in the area of parallel and distributed computing, are locally checkable, where the correctness of the whole solution can be verified by checking the partial solution around the local neighborhood of each node.
In particular, the class of \lcl problems consists of problems with a finite set of outputs per node/edge and a finite set of locally feasible solutions (see \Cref{def:lcl}), and includes fundamental problems such as MIS, node/edge-coloring and the algorithmic Lov\'{a}sz Local Lemma (\llle).
Our work shows that any \lcl problem can be solved in $O(\log D)$ time and that the same runtime can be obtained for many problems that are not restricted to finite descriptions.
We complement our results by showing that for \lcl problems, this bound is tight under the \vs conjecture. 

In recent works, the complexity of \lcls in \mpc was compared against \emph{locality}~\cite{mpchierarchy,mpc-landscape}, where locality refers to the round complexity of solving an \lcl in the \local model, as a function of $n$.
It was shown that all \lcls on trees can be solved exponentially faster in \mpc as compared to \local. As a consequence, all \lcls on trees can be solved in $O(\log n)$ rounds in the low-regime \mpc model. We note that it is often the case that  the diameter of a graph is small, potentially much smaller than the locality of a certain graph problem (which is \emph{independent} of the diameter).
Hence, our novel technique significantly improves on the state-of-the-art runtimes for various graph problems in a broad family of graphs.

\subsection{Our Contributions}
Our main contribution is an algorithm that deterministically detects the connected components of a forest in time logarithmic on the maximum component-wise diameter; crucially, independent of the size $n$ of the input graph, whose dependence is inherently present in the techniques used in previous works. We also show that our approach is asymptotically optimal under the \vs conjecture. Next, we present our results more formally.

\begin{restatable}[Connected Components]{theorem}{thmCC}
	\label{thm:CCMainTheorem}
	Consider the family of forests.
	There is a deterministic low-space \mpc algorithm to detect the connected components on this family of graphs. 
	In particular, each node learns the maximum ID of its component.
	The algorithms runs in $O(\log D)$ rounds, where $D$ is the maximum diameter of any component.
	The algorithm requires $O(n+m)$ words of global memory, it is component-stable, and it does not need to know $D$.
	Under the \vs conjecture, the runtime is asymptotically optimal.
\end{restatable}

The techniques for \Cref{thm:CCMainTheorem} can be extended to also obtain a rooted forest, where each node also knows the ID of the corresponding root.

\begin{restatable}[Rooting]{theorem}{thmTreeRooting}
	\label{thm:RootingMainTheorem}
	Consider the family of forests with component-wise maximum diameter $D$. There is a deterministic low-space \mpc algorithm that roots the forest in $O(\log D)$ rounds using $O(n+m)$ words of global memory, and it is component-stable.
\end{restatable}

The rooting of the input forest gives us a handle for easier algorithm design and memory allocation in low-space \mpc.
As a concrete example, our results yield an $O(\log D)$ algorithm for deterministically 2-coloring forests.
Without going into technical details, this can be achieved through a rather simple algorithm, where each node decides its color based on the parity of its distance to the root, and only needs to keep one pointer in memory for the parity counting.
In a sense, we outsource the tedious implementation details to the rooting algorithm in \Cref{thm:RootingMainTheorem} and obtain a convenient tool for algorithm design.

More broadly, we show how to solve any \lcl problem in rooted forests in $O(\log D)$ deterministic rounds. \lcl{}s have gotten ample attention in various distributed models of computation, e.g., \cite{mpchierarchy,BCMOS21, BHKLOS18, ChangKP19, CP19}.  Roughly speaking, the family of \lcl{}s is a subset of the problems for which we can check if a given solution is correct by inspecting the constant radius neighborhood of each node. (see \Cref{def:lcl} for a formal definition of \lcls).
Furthermore, we show that for any fixed $D \in \Omega(\log n)$ and $D \in n^{o(1)}$, there cannot exist an algorithm that solves all \lcl problems in time $o(\log D)$ in the family of unrooted forests of diameter at most $D$.
This holds even if $\poly(n)$ global memory is allowed.

\begin{theorem}[\lcls on trees, simplified]
\label{thm:LCLSolver}
    Consider an \lcl problem $\Pi$ on forests and let $D$ be the component-wise maximum diameter.
	There is a deterministic low-space \mpc algorithm that solves $\Pi$ in $O(\log D)$ rounds using $O(n+m)$ words of global memory.
	Under the \vs conjecture, the runtime is asymptotically optimal.
\end{theorem}

\subsection{Challenges and Techniques}
\label{ssec:challenges}
A canonical approach to solve connected components on forests is to root each tree and identify each tree with the ID of the root.
Also, examining the challenges in rooting demonstrates the challenges we face when identifying connected components.
A natural approach to root a tree is to iteratively perform \emph{rake} operations, i.e., pick all the leaves of the tree and each leaf picks the unique neighbor as its parent.
This approach clearly roots a tree in $O(D)$ parallel rounds and furthermore, in the case of a forest, each tree performs its rooting process independently.
If we ignore the memory considerations in the low-space \mpc model, this process could be implemented in $O(\log D)$ rounds using the graph exponentiation technique, where, in $O(\log D)$ rounds,  \emph{every} node gathers their $D$-hop neighborhoods, i.e., the whole graph, to simulate the process fast.
However, when we limit the global memory to $O(n+m)$, we get into trouble.
A simulation through graph exponentiation requires that \emph{all} nodes iteratively gather larger and larger neighborhoods \emph{simultaneously}.
With the strict memory bound, this implies that a node can only gather a \emph{constant} radius neighborhood (and in non-constant degree graphs that we deal with even that is not possible!), which allows only for simulating a constant number of rake-iterations in one \mpc round.

A hope towards a more efficient approach is to show that the amount of total memory \emph{relative}  to the nodes remaining in the graph increases as we rake the graph (similarly to previous work~\cite{Andoni2018, Behnezhad2019, componentstable}).
If one can reduce the size of the graph by a constant factor in each \mpc round, then the available total memory increases by a constant factor per remaining node.
Then, we can gather a slightly larger neighborhood per node in the next step of the simulation.
However, even if we had this guarantee, the best we could hope for is a runtime that depends on $n$, since this approach relies on a progress measure that depends on shrinking the graph.
Informally speaking, this observation says that we need to have a fundamentally different approach than gradually sparsifying the graph.

\paragraph{Balanced Exponentiation.}
One of our main technical contributions is to introduce a new method to gather a \emph{part} of the neighborhood of each node that is balanced in the following sense.
Suppose, for the sake of argument, that we have a rooted tree.
Then, if a node $u$ has, say, $\gamma$ descendants, we ensure that $u$ will  only gather $O(\gamma)$ nodes in the direction of the root, i.e., the direction opposing its descendants.
Furthermore, it will also gather its $\gamma$ descendants, resulting in a memory demand of $O(\gamma)$ (for $u$).
A crucial step in our analysis is to show that even if each node gathered their $\gamma$ descendants and $O(\gamma)$ nodes in the direction of the root, we do not create too much redundancy and we respect the linear total memory bound.
A key technical challenge here is that there is no way for a node to know who are its descendants (because the input graph is unrooted).
We show that, without an asymptotic loss in the runtime, we can deterministically determine which neighbor of $u$ is the \emph{worst case} for a choice of a parent and gather the respective nodes slower.

\paragraph{Progress Measure.}
As mentioned above, to obtain a runtime independent of $n$, we need to avoid arguments that are based on the size of the graph getting smaller during the execution of our algorithm.
The topology gathering through exponentiation can be seen as creating a virtual graph, where a virtual edge $\{u, v\}$ corresponds to the fact that $u$ knows how to reach $v$ and vise versa.
The base of our progress measure is to aim to show that in this virtual graph, the diameter is reduced by a constant factor in each iteration.
Unfortunately, having this type of guarantee seems too good to be true.
Already on a path, it requires too much memory to create a virtual graph where the distances between all pairs of nodes are reduced.
Our contribution is to show that this example is degenerate in the sense that either we can guarantee that the balanced exponentiation reduces the diameter or we can reduce it through a node-contraction type of operation.

\subsection{Further Related Work}
In relation to our work, previous works have studied finding rooted spanning forests.
In~\cite{sirocco, Andoni2018}, $O(\log D \cdot \log \log n)$ algorithms for rooting were given and the runtime was improved to $O(\log D + \log \log n)$ by~\cite{ccderandom}.

Locally checkable problems have been intensively studied in the \mpc model.
Many classic algorithms from PRAM imply \mpc algorithms with the same runtime, e.g., the MIS, maximal matching and coloring~\cite{luby86, alon86}.
The runtime of such simulations are typically polylogarithmic and, in \mpc, the aim is to obtain something significantly faster.
For MIS and maximal matching, there are $\widetilde{O}(\sqrt{\log \Delta} + \log \log \log n)$ time algorithms~\cite{GU19} and $(\Delta + 1)$-node-coloring can be solved in $O(\log \log \log n)$ rounds, even deterministically~\cite{Chang2019, componentstable}.

Many of the current state-of-the-art algorithms for locally checkable problems are (at least to some degree) based on distributed message-passing algorithms.
The common design pattern is to design a message-passing algorithm, for example in the \local model of distributed computing~\cite{Linial92} where the output of each node is decided according to their $t$-hop neighborhood in $t$-rounds.
These algorithms are then implemented faster in the \mpc model through the graph exponentiation technique~\cite{wattenhofer} that, in the ideal case, collects the $t$-hop ball around each node in $O(\log t)$-rounds. This framework was used to obtain an exponential speedup for many locally checkable problems in general, and in particular,
it was used recently to show that all \lcl problems on trees with $t$-round complexity in \local can be solved in $O(\log t)$ \mpc rounds~\cite{mpchierarchy,mpc-landscape}.
Our work broadens our understanding on the complexities of \lcls as a function of the diameter, which is somewhat orthogonal to previous works.

On a technical level, a related work gave a clever approach to encode the feasible outputs around each node into a constant sized \emph{type} of the node~\cite{CP19}.
Given a rooted tree, the type of a node $u$ (or its subtree) is determined through the set of possible outputs of its descendants.
This encoding gives rise to an efficient convergecast protocol, where the root learns its type and effectively broadcasts a valid global solution to the rest of the tree.
In a recent work, related techniques were used to implement a message passing algorithm for \lcls on trees using small messages~\cite{BCMOS21}.
In our work, we employ similar ideas to aggregate and broadcast information efficiently through the input tree.

\paragraph{Lower Bounds.}
While simulating \local message passing algorithms in \mpc has been fruitful in algorithm design, there is an inherent limitation to this approach.
A na\"{i}ve implementation results in a \emph{component-stable} algorithm, where we can show that the simulation cannot be more than exponentially faster than the message passing algorithm~\cite{Ghaffari2019, componentstable}.
An algorithm is said to be \emph{component-stable} if the output on node $v$ depends (deterministically) only on the topology, the input of the nodes, and the IDs of the nodes in the connected component of $v$.
Furthermore, the output is allowed to depend on the number of nodes $n$, the maximum degree $\Delta$ of the input graph, and in case of randomized algorithms, the output can depend on shared randomness.
It was shown that for component-stable algorithms and under the \vs conjecture, $\Omega(\log t)$ rounds cannot be beaten if $t$ is a lower bound on the complexity of the given problem in the \local model.
We emphasize that our lower bounds also work for component-unstable algorithms (still relying on the \vs conjecture).

\section{Overview, Roadmap and Notation}
Our formal results are presented in \Cref{sec:connectedComponents,sec:ConnectedComponentsTheReal,sec:LCLSection,sec:Hardness}.  In \Cref{sec:highlevel}, we present the core techniques of our algorithm. The formal version of this algorithm appears in \Cref{sec:connectedComponents}.

\paragraph{\Cref{sec:connectedComponents}:} This section contains the most involved part of our work, i.e., an algorithm that lets every node of an input tree output the maximum identifier of the tree. On a tree $G$, our algorithm runs in $O(\log \ad)$ rounds and uses global memory $O((n+m)\cdot \ad^3)$  where the parameter $\ad\in  [\diam(G), n^{\delta/8}]$ needs to be known to the algorithm. This sounds like a foolish approach, as this problem can be trivially solved in $O(1)$ rounds if we were really given a single tree as input. We still chose to present our result in this way, as our seemingly na\"ive algorithm is the core of our connected components algorithm that we present in \Cref{sec:ConnectedComponentsTheReal}.

\paragraph{\Cref{sec:ConnectedComponentsTheReal}:}  In fact,  we show that our algorithm can be correctly extended to forests. Note that if every node knows the maximum identifier of its tree, we automatically solve the connected components problem. In this section, we also show how to remove the requirement of knowing $\ad$ via doubly exponentially increasing guesses for $\ad$. 
We also show that we can reduce the overall memory requirement to $O(n+m)$ by preprocessing the graph, that is, we spend additional $O(\log \ad)$ rounds to reduce the size of the graph by a factor $\ad^3$. 
Lastly, we show that the runtime reduces to $O(\log \max_i\{ D_i \})$ where $D_i$ is the diameter of the $i$-th component of the input graph. 
In this section,  we also present the full proof of our connected components algorithm (\Cref{thm:CCMainTheorem}) and our rooting algorithm (\Cref{thm:RootingMainTheorem}).

\paragraph{\Cref{sec:LCLSection}:}
In this section, we show a nice application of our rooting algorithm from \Cref{thm:RootingMainTheorem}. In particular, we show that any \lcl problem can be solved in just $O(\log D)$ rounds, once each tree of the forest is rooted (\Cref{thm:LCLSolver}).

Our approach has a dynamic programming flavor and we explain it for a single tree of the forest. We iteratively reduce the size of the tree, by compressing small subtrees into single nodes, and paths into single edges. 
While performing these compressions, we set additional constraints on the solution allowed on the nodes into which we compress subtrees, and on the edges that represent compressed paths. We maintain the invariant that, if we obtain a solution in the smaller tree, then it can be extended to the original one. We show that, by performing a constant number of compression steps, we obtain a tree that is comprised of a single node, where it is straightforward to compute a solution. We then perform the same operations in the reverse order, in order to extend the solution to the whole tree. All of this is preceded by using \Cref{thm:RootingMainTheorem} to compute a rooting of the tree/forest. The rooting helps, as with a given rooting it is significantly easier to identify the suitable subtrees to compress without breaking memory bounds. 

\paragraph{\Cref{sec:Hardness}:}
In this section, we show that the runtimes of our algorithms are tight, conditioned on the widely believed \vs conjecture.
Our aim is to use a reduction from the \vs problem to solving connectivity on paths.
In previous work~\cite{Ghaffari2019}, a reduction to connectivity on paths was introduced, but for technical reasons, it is not sufficient for our purposes.
We require a guarantee that each path is of bounded diameter, which is not directly guaranteed by the previous work.
Hence, we start by defining a problem on forests, called \diamConnectivity, for which we can prove conditional hardness. By a reduction, we obtain a conditional lower bound of $\Omega(\log D)$ for the connected components problem.

We then define an \lcl{} problem such that, given an algorithm for it, we can use it to solve the \diamConnectivity problem. Hence, we obtain a lower bound of $\Omega(\log D)$ for the problem, implying that our generic \lcl{} solver is also conditionally tight.

\subsection{Definitions and Notation}
\label{ssec:definitionNotation}
Given a graph $G = (V,E)$, we denote with $\Delta$ the maximum degree of $G$, with $n = |V|$ the number of nodes in $G$, and with $m=|E|$ the number of edges in $G$. We denote with $N_G(v)$ the neighbors of $v$, that is, the set $\{u \mid \{u,v\} \in E\}$. We denote with $\deg_G(v)$ the degree of a node $v$, that is, the number of neighbors of $v$ in $G$. If $G$ is clear from the context, we may omit $G$ and simply write $N(v)$ and $\deg(v)$. If $G$ is a directed graph, we denote with $\deg(v)$ the degree of $v$ in the undirected version of $G$, and with $\deg_{\text{in}}(v)$ and $\deg_{\text{out}}(v)$ its indegree and its outdegree, respectively.
We define $\dist_G(u,v)$ as the hop-distance between $u$ and $v$ in $G$. Again, we may omit $G$ if it is clear from the context. The \emph{radius-$r$ neighborhood} of a node $v$ is the subgraph $G_r(v)=(V_{r}(v), E_{r}(v))$, where $V_{r}(v)=\{u\in V~:~ \dist(u,v)\le r \}$, and $E_{r}(v)=\{(u,w)\in E~:~ \dist(v,u)\le r \mbox{ and } \dist(v,w)\le r \}$. Also, we denote with $G^k$ the $k$-th power of $G$, that is, a graph containing the same nodes of $G$, where we connect two nodes $u$ and $v$ ($u \neq v)$ if and only if they satisfy $\dist_G(u,v) \le k$. The \emph{eccentricity} of a node $v$ in a graph $G$ is the maximum of $\{\dist(u,v) ~|~ u \in V\}$.

\section{The \maxid Problem: Overview and Techniques} \label{sec:highlevel}

In this section, we present the core techniques of our specialized algorithm for solving the \maxid problem, that is the core ingredient for solving connected components on forests (upper bound of \Cref{thm:CCMainTheorem}). In the \maxid problem, one is given a  connected tree with a unique identifier for each node, and all nodes must output the maximum identifier in the tree. We note that it is trivial to solve the problem in $O(1)$ \mpc rounds using a broadcast tree; however, this approach does not extend to forests, and hence a more sophisticated solution is required. The purpose of this section is to present the high level ideas of an algorithm that solves \maxid and can also be extended to forests. Some lemma statements have been adapted to fit this (informal) version.

\medskip

\noindent\textbf{\Cref{lem:maxid}} (Solving \maxid on trees)\textbf{.} 
\emph{Consider the family of trees. There is a deterministic low-space \mpc algorithm that solves \maxid on any graph $G$ of that graph family when given $\ad \in [\diam(G),n^{\delta/8}]$. The algorithms runs in $O(\log \ad)$ rounds, is component-stable\footnote{By the formulation of \Cref{def:maxIDproblem}, any algorithm solving \maxid is component-stable by definition. This is discussed in detail in \Cref{ssec:ccForestsAndStability}}, and requires $O(m \cdot \ad^3)$ words of global memory.}

\medskip

We begin with definitions that are essential not only to define our algorithm but also for proving its memory bounds. 
Let $v$ be a vertex of a tree $G$. 
For all nodes $u\in N(v)$, define
\[G_{v \arr u}=\{w\in V(G)\mid \text{$u$  is contained in the shortest path from $v$ to $w$}\}\]
to be all nodes in the tree that are reachable from $v$ via $u$, including $u$. Also, let $G_{v \narr u} \coloneqq V(G) \setminus G_{v \arr u}$.  For every $w \in G$, let $r_v(w)$ be the node $u \in N(v)$ satisfying that $w \in G_{v \arr u}$, i.e., $r_v(w)$ is the neighbor of $v$ which is on the unique path from $v$ to $w$.

\begin{definition*}[Light and heavy nodes]
	Let $0<\delta<1$ be a constant. A node $v$ is \emph{light against} a neighbor  $u \in N(v)$ if $|G_{v \narr u}| \leq n^{\delta/8}$.
	A node is \emph{light} if it is light against at least one of its neighbors. 
	Nodes that are not light are \emph{heavy}.  
\end{definition*}

\begin{restatable}{figure}{exampleFig}
	\centering
	\includegraphics[width=0.7\textwidth]{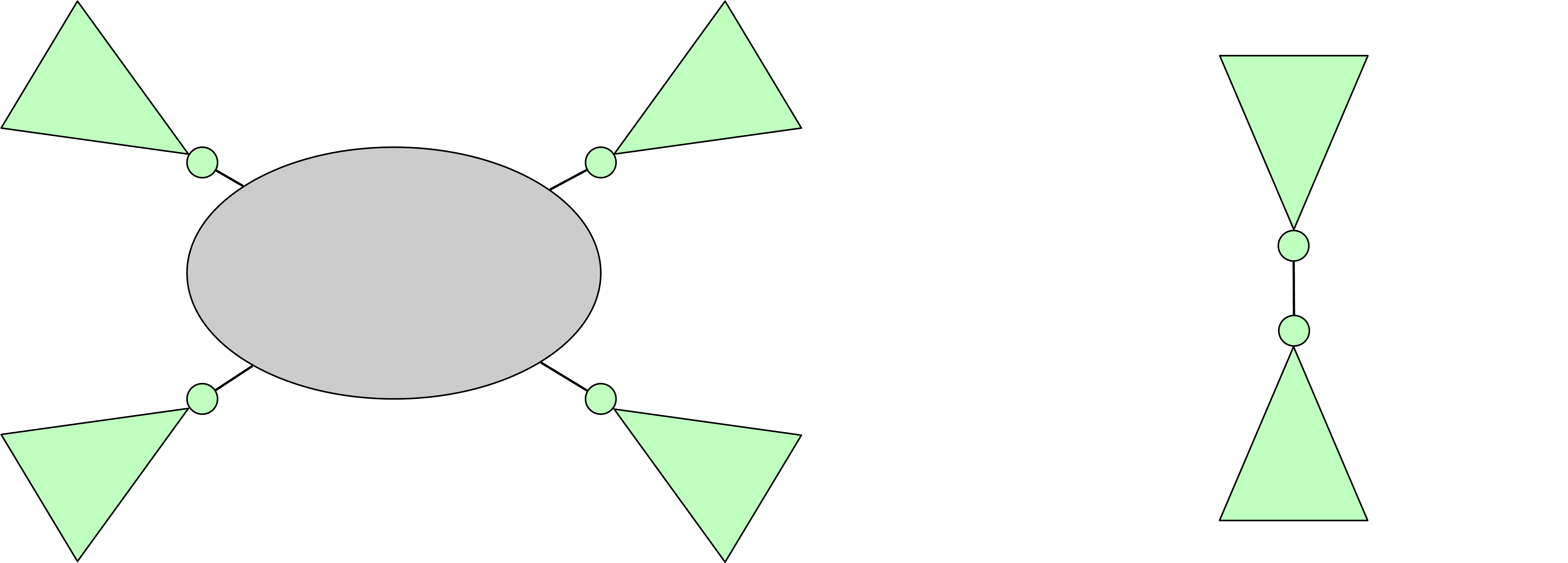}
	\caption{Light nodes are green and heavy nodes are gray.}
	\label{fig:exampleGraphs}
\end{restatable}

If there are no heavy nodes,  the graph is small and fits into the local memory of one machine.

\begin{lemma*}[see \Cref{cor:sizeOfGraph}]
	Any tree with no heavy nodes contains at most $2n^{\delta/2}$ vertices. 
\end{lemma*}

We prove that as soon as there is at least one heavy node, the graph has to look like the one depicted on the left hand side of \Cref{fig:exampleGraphs}, that is, light subtrees that are attached to a connected component of heavy nodes. We exploit this structure in our algorithm. 

\subsection{\maxid: The Algorithm}

The high-level idea is to iteratively compress parts of the graph (without disconnecting it) such that the knowledge of  the maximum identifier of the compressed parts is always kept within the resulting graph. We repeat this process until there remains only one node, that knows \ID, the maximum identifier in the graph. Then, we backtrack the process by iteratively decompressing and broadcasting the knowledge about \ID. Eventually, we are left with the original graph where all nodes know \ID. 

More in detail, our algorithm consists of $\ell=O(1)$ \emph{phases} and the same number of \emph{reversal phases}. During the phases, we first compress all light subtrees into single nodes (a procedure that we refer to as $\clst$) and then replace all paths by a single edge ($\cp$). We denote the resulting graphs by $G_0,G_1$,\ldots,$G_{\ell}$. 
The phases are followed by reversal phases, in which we undo all compression steps of the regular phases in reverse order to spread \ID to the whole graph.

\paragraph{Bounding the number of Phases.}
Consider some phase $i$ and graph $G_i$ with heavy nodes that looks as illustrated in \Cref{fig:exampleGraphs} (left). If we remove all light subtrees from graph $G_i$, for the resulting graph $G_{i+1}$, it holds that every leaf (aka a formerly heavy node) corresponds to a distinct removed subtree of size at least $n^{\delta/8}$ (if the subtree was smaller the leaf would not be heavy). If we then contract all paths in $G_{i+1}$ into single edges, leaving no degree-$2$ nodes in $G_{i+1}$, it holds that at least half of the nodes in $G_{i+1}$ corresponds to a removed subtree. As each of these subtrees has $\geq n^{\delta/8}$ distinct nodes,   we have removed a polynomial-in-$n$ fraction of nodes from $G_i$ to obtain $G_{i+1}$. Hence, we can only repeat the process a constant time until the graph becomes small.

\subsection{\maxid: Compressing Light Subtrees}

For the sake of this high level overview, we focus on our most involved part, the procedure that compresses (maximal) light subtrees into the adjacent heavy node (\clst). The difficulty is that nodes do not know whether they are light or heavy, and already one single exponentiation step in the ``wrong'' direction of the graph can ruin local and global memory bounds. However, there seems to be no way to obtain a runtime that is logarithmic in the diameter without exponentiation. Thus, we perform careful exponentiations that always ensure the memory bounds but at the same time make enough progress. 

Consider a graph $G$ with $n$ nodes---starting from the second phase we will actually use this algorithm on graphs with fewer than $n$ nodes. At all times, every node $v$ has some set of nodes $S_v$ in its memory, which we initialize to $N(v)$. Set $S_v$ can be thought of as the node's view or knowledge. During the execution, $S_v$ grows, and if $|S_v| \geq 2n^{\delta/4}$, $v$ becomes \full. Similarly to definitions $G_{v \arr u}$ and $G_{v \narr u}$, let us define the following. For a node $v$ and a node $u \in N(v)$, let $S_{v \arr u}=S_v\cap G_{v\arr u}$. Also, let $S_{v \narr u} \coloneqq S_v \setminus S_{v \arr u}$. 

All nodes in the graph have the property that they are either light or heavy. Initially, nodes themselves do not  know whether they are light or heavy, since these properties depend on the topology of the graph. During the algorithm each node is in one of the four states: \act, \happy, \full, or \sad. Initially, all nodes are \act. A node $v$ becomes \happy, if at some point during the execution, there exists $u \in N(v)$ such that $G_{v \narr u}\subseteq S_{v}$ and $|G_{v \narr u}| \leq n^{\delta/8}$.  If a node, that is not \full, realizes that it can never become \happy (for example by having $|S_{v \arr u}|>n^{\delta/8}$ for two different neighbors $u$), it becomes \sad. Upon becoming \happy, \sad or \full, nodes do not partake in the algorithm except for answering queries from \act nodes. We call nodes \unhappy if they are in some other state than \happy (including state \act). The goal is that all light nodes become \happy. We will prove that the algorithm that we will provide satisfies the following lemma.  
\begin{restatable*}{lemma}{lemCorrectnessLemmaLight} \label{lem:correctnessLemmaLight}
	After $O(\log \ad)$ iterations, all light nodes become \happy, while heavy nodes always remain \unhappy.
\end{restatable*}
Intuition for its correctness requires further details and is defered to the end of this section.

When comparing the definitions of \happy and light, it is evident that when a node becomes \happy, it knows that it is light. Similarly, a node becoming \full or \sad knows that it is heavy. 
At the end of the algorithm, \happy nodes with a \full or \sad neighbor compress their whole subtree in that neighbor. A crucial challenge here is to ensure that these compressions are not conflicting as all such nodes execute these in parallel and without a global view. 

\paragraph{Exponentiation.} Recall the definition of $r_v(w)$ at the beginning of the section. For a node $v$ and any $X\subseteq N(v)$, define an exponentiation operation as
\begin{align*}
	\expo(X): ~~ S_v \larr \bigcup_{u \in X} \bigcup_{w \in S_{v \arr u}} S_{w \narr r_w(v)}. 
\end{align*}

We say that node $v$ \emph{exponentiates  in the direction of $u\in N(v)$} if $v$ performs $\expo(X)$ with $u\in X$. 

The algorithm consists of $O(\log \ad)$ iterations, in each of which nodes perform a carefully designed graph exponentiation procedure. The aim is for light nodes $v$ to become \happy by learning their subtrees, after which, (certain) light nodes compress into their \unhappy neighbor.

\paragraph{Failed Exponentiation Approaches.}
If there were no memory constraints and every node could do a proper (uniform) exponentiation step in every iteration of the algorithm, i.e., execute $\expo(N(v))$, after $O(\log \ad)$ iterations \emph{all} nodes would learn the whole graph---a proper exponentiation step executed on all nodes halves the diameter---and the highest ID node could compress the whole graph into itself. However, uniform exponentiation would result in all nodes exceeding their local memory $O(n^\delta)$, and also significantly breaking the global memory requirement.
 Even if we were to steer the exponentiation procedure such that light nodes would learn a $D_v$ radius ball around them, where $D_v$ is the diameter of their light subtree, this would still break global memory. In fact, we cannot even do a single exponentiation step for all nodes in the graph without breaking memory bounds!

\paragraph{Our Solution (Careful Exponentiation \& Probing).} Hence, we let every node exponentiate in all but one direction, sparing the direction which currently looks most likely to be towards the heavy parts of the graph. Note that the knowledge of a node about the tree changes over time and in different iterations it may spare different directions. This step is further complicated as nodes neither know whether they are heavy or light nor do they know the size of their subtree, nor in which direction the heavy parts of the graph lie. Thus, in our algorithm, nodes perform a careful probing for the number of nodes into all directions to determine in which directions they can safely exponentiate without using too much memory. More formally, a node $v$ computes 
$B_{v \arr u} = \sum_{w \in S_{v \arr u}} |S_{w \narr r_w(v)}|$  for every neighbor $u \in N(v)$ as an estimate for the number of nodes it may learn when exponentiating towards $u$. This estimate may be inaccurate and may contain a lot of doublecounting. The precise guarantees that this probing provides are technical and presented in \Cref{sec:CC}. 

We now reason (in a nutshell) why this algorithm meets the memory requirements and why we still make enough progress in order to make all light nodes \happy in $O(\log \ad)$ iterations.

\paragraph{Local and Global Memory Bounds.}

If there were no memory limitation, we would already know that  after $\ell=O(1)$ phases $G_{\ell}$ would consist of a single node $r$. For the sake of analysis, we assume a rooting of $G$ at $r$. We emphasize that fixing a rooting is only for analysis sake, and we do not assume that the tree is actually rooted beforehand.

Given the rooting at $r$, we define $T(v,r)$ as the subtree rooted at $v$ (including $v$ itself). Then, the following lemma is crucial to bound the memory. 
The lemma is standalone as it does not use any properties of $r$.

\begin{restatable*}{lemma}{lemglobalMemory} \label{lem:globalmemory}
	Consider an $n$-node tree $T$ with diameter $D$ that is \emph{rooted} at node $r$. Let $T(v,r)$ denote the subtree rooted at $v$ (including $v$) when $T$ is rooted at $r$. It holds that $\sum_{v \in V} |T(v,r)| \leq (D+1) \cdot n$.
\end{restatable*}

\begin{proof}
	Consider the unique path $P_{rv}$ from the root $r$ to a node $v$. Observe that node $v$ is only in the subtrees of the nodes in $P_{rv}$. Since $|P_{rv}| \leq D+1$, node $v$ is overcounted at most $D$ times, and $\sum_{v \in V} |T(v,r)| \leq (D+1) \cdot n$.
\end{proof}

The probing ensures that a node, if it exponentiates into a direction, essentially never learns more nodes than there are contained in its ``rooted subtree'' $T(v,r)$.  
\begin{lemma*}[see \Cref{lem:memoryDirectionExponentiation}] 
	Let $v$ be any node with a parent $u$ (according to the hypothetical rooting at $r$). If in some iteration, node $v$ exponentiates in the direction of $u$, i.e., it performs $\expo(X)$ with $u \in X$, the size of the resulting set $S_{v \arr u}$ is bounded by $|T(v,r)| \cdot \ad$.
\end{lemma*}
This is sufficient to sketch the global memory bound. 
\begin{lemma*}[see \Cref{lem:rootingGlobalMemory}]
	In \clst, the global memory never exceeds $O(n \cdot \ad^3)$.
\end{lemma*}

\begin{proof}[Proof sketch]
	Assume that there is at least one heavy node and consider an arbitrary iteration $j$ of the algorithm. For node $v$ define the set $C_v \subseteq S_v$ as the set of nodes that $v$ has added to $S_v$ as a result of performing $\expo$ in all iterations up to iteration $j$. Let $u$ be the parent of $v$ (according to the hypothetical rooting at $r$). For that $u$, let $C_{v \arr u} \coloneqq C_v\cap G_{v\arr u}$. We obtain
	\begin{align*}
		|C_v| &\leq |C_{v \arr u}| + \sum_{w \in N(v) \setminus u} |S_{v \arr w}| \leq |T(v,r)|\cdot \ad + |T(v,r)| = (1 + \ad) |T(v,r)|~.
	\end{align*}
	
	The bound on $|C_{v \arr u}|$ is obtained by applying the previous lemma for the last iteration where $v$ has exponentiated in the direction of $u$, and the bound on the sum is by the definition of $T(v,r)$.
	
	We need to introduce the notation $C_v$, as in our actual algorithm, exponentiations are not symmetric. In order to ensure a \emph{symmetric enough} view, nodes $v$ that add some vertex $u$ to their set $S_v$ also add themselves to the set $S_u$. Thus $C_v\neq S_v$. However, this results in at most a factor $2$ increase in global memory. The total memory is then bounded by

	\begin{align*}
		 \sum_{v \in V} |S_v| = \sum_{v \in V}  2|C_v| \leq \sum_{v \in V}  2(1+\ad) |T(v,r)| = O(n \cdot \ad^2).
	\end{align*}
	Here, the bound on $\sum_{v \in V} |T(v,r)|$ is due to \Cref{lem:globalmemory}.
	The additional $\ad$ factor in the lemma statement is due to the fact that a node may learn about the same node $\ad$ times in a single exponentiation step resulting in a local peak in global memory; details are given in the full proof. 
\end{proof}

The bounds on local memory use that the probing ensures that we do not exponentiate into a direction if it would provide us with too many new nodes.

\paragraph{Measure of Progress.}
In order to show that all light nodes become happy, we prove that the distance between a light node and a leaf in its subtree decreases by a constant fraction in a constant number of rounds. Distance, in this case, can be measured via a virtual graph where there is an edge between two nodes $u$ and $v$ if $v\in S_u$ or $u\in S_v$. Our algorithm design ensures that light nodes always exponentiate in all but one direction. This is sufficient to show that each segment $x_1,\ldots x_5$ of length $5$ of a shortest path in the virtual graph $H$, shortens by at least one edge in each iteration. Intuitively, one can simply use that $x_3$ in such a segment either exponentiates in the direction of $x_1$ or $x_5$ and will hence add the respective node to its memory. The actual proof needs a more careful reasoning, e.g., as we cannot rely on $x_1$ being part of the memory of $x_2$, due to non homogeneous exponentiations in previous iterations.

\section{The \maxid Problem}

\label{sec:connectedComponents}

In this section, we give a specialized algorithm for the \maxid problem on trees, which will be the core ingredient for solving connected components on forests (upper bound of \Cref{thm:CCMainTheorem}). Once having the algorithm for solving \maxid, one can extend it to work as a connected components algorithm. We defer this extension and its proofs to \Cref{sec:CC}. We define the problem as follows.

\begin{definition}[The \maxid problem] \label{def:maxIDproblem}
    Given a connected graph with a unique identifier for each node, all nodes output its maximum identifier.
\end{definition}

\begin{lemma}[Solving \maxid on trees] \label{lem:maxid}
    Consider the family of trees. There is a deterministic low-space \mpc algorithm that solves \maxid on any graph $G$ of that graph family when given $\ad \in [\diam(G),n^{\delta/8}]$. The algorithms runs in $O(\log \ad)$ rounds, is component-stable\footnote{By the formulation of \Cref{def:maxIDproblem}, any algorithm solving \maxid is component-stable by definition. This is discussed in detail in \Cref{ssec:ccForestsAndStability}}, and requires $O(m \cdot \ad^3)$ words of global memory.
\end{lemma}

\subsection{Definition and Structural Results} \label{ssec:ccDefs}
We begin with structural properties of trees that are essential for proving our memory bounds. Also, the introduced notation plays a central role in each step of our algorithms. 
Let $v$ be a vertex of a tree $G$. 
For all nodes $u\in N(v)$, define
 \[G_{v \arr u}=\{w\in V(G)\mid \text{$u$  is contained in the shortest path from $v$ to $w$}\}\]
  to be all nodes in the tree that are reachable from $v$ via $u$, including $u$. Also, let $G_{v \narr u} \coloneqq V(G) \setminus G_{v \arr u}$.  For every $w \in G$ let $r_v(w)$ be $u \in N(v)$ such that $w \in G_{v \arr u}$, i.e., $r_v(w)$ is the neighbor of $v$ which is on the unique path from $v$ to $w$.

\begin{definition}[Light and heavy nodes] \label{def:lightheavy}
	Let $0<\delta<1$ be a constant. A node $v$ is \emph{light against} a neighbor  $u \in N(v)$ if $|G_{v \narr u}| \leq n^{\delta/8}$.
	A node is \emph{light} if it is light against at least one of its neighbors. 
	When $v$ is light against $u$, let $T_{v,u}$ denote $G_{v \narr u}$. Nodes that are not light are \emph{heavy}. 
\end{definition}

Observe that a light node $v$ can be light against multiple neighbors $u$ and hence, we need to use a subscript in the notation $T_{v,u}$. We emphasize that $v \in T_{v,u}$. Throughout most of our proofs we need to consider the cases that a (virtual) tree contains heavy nodes and the case that it only consists of light nodes separately. Both situations are depicted in \Cref{fig:exampleGraphs}. We continue with proving structural properties for both cases. Any tree has light nodes as its leaves are light.

\begin{lemma} \label{lem:allLight} Consider a tree $G$ that contains a heavy node and let $v\in G$ be a light node against neighbor $u$. Then all nodes $x \in T_{v,u}$ are light. Moreover, any $x\in T_{v,y}$, $x\neq v$ is light against $r_x(v)$.

\end{lemma}

\begin{proof}
	The first part of the claim must holds since $G_{x \narr r_x(v)} < G_{v \narr u} \leq n^{\delta/8}$. The second part must hold, since otherwise, all nodes are light, contradicting the  assumption that there exists a heavy node.
\end{proof}

\begin{lemma} \label{obs:heavyCCandTVunambiguous} For every tree $G$ with at least one heavy node, it holds that (i) heavy nodes induce a connected component, and that (ii) every light node is light against exactly one neighbor. 
\end{lemma}

\begin{proof}
	For both parts, assume the opposite. Then there is a heavy node in $T_{v,u}$ for some light node $v$, which contradicts \Cref{lem:allLight}.
\end{proof}

Due to \Cref{obs:heavyCCandTVunambiguous}, we write $T_v$ instead of $T_{v,u}$ for a light node in a tree with (a) heavy node(s) and call $T_v$ the node's \emph{subtree}. 

\begin{observation}
\label{lem:newLemmaThreeFour}
For any tree $G$ and any two adjacent nodes $u, v\in G$, we have $G=G_{v\narr u}\cup G_{u\narr v}$.
\end{observation}

\begin{proof}
	Since $G_{v \arr u} = G_{u \narr v}$, we obtain $G_{u \narr v} \cup G_{v \narr u}= G_{v \arr u} \cup G_{v \narr u}=G$.
\end{proof}

\begin{lemma} \label{cor:sizeOfGraph}
	Any tree with no heavy nodes contains at most $2n^{\delta/2}$ vertices. 
\end{lemma}
\begin{proof}
For the sake of analysis, let each node $v$ put one token on each incident edge $\{v,u\}$ where $v$ is light against $u$, i.e.,  $|G_{v\narr u}|\leq n^{\delta/8}$.  As all nodes are light the total number of tokens is at least as large as the number of nodes. Since the graph is a tree, at least one edge receives two tokens. Let $\{u,v\}$ be such an edge and observe that $G=G_{v\narr u}\cup G_{u\narr v}$ holds due to \Cref{lem:newLemmaThreeFour}. It holds that $|G|=|G_{v\narr u}\cup G_{u\narr v}|\leq 2n^{\delta/8}$, because $v$ is light against $u$ and $u$ is light against $v$. 
\end{proof}

The following lemma will be central to bounding the global memory of our algorithm. It considers a rooted tree, which we will only use for analysis; we do not assume a rooting is given as input.

\lemglobalMemory

\begin{proof}
	Consider the unique path $P_{rv}$ from the root $r$ to a node $v$. Observe that node $v$ is only in the subtrees of the nodes in $P_{rv}$. Since $|P_{rv}| \leq D+1$, node $v$ is overcounted at most $D$ times, and $\sum_{v \in V} |T(v,r)| \leq (D+1) \cdot n$.
\end{proof}

\subsection{\maxid: The Algorithm} \label{ssec:ccTheAlgorithm}

In this section, we present a \maxid algorithm for trees, which we refer to as \cc. In our algorithm, every node of an input tree $G$ outputs the maximum identifier of the tree, which we denote by \ID. We assume we are given $\ad\in  [\diam(G), n^{\delta/8}]$. The runtime of our algorithm is $O(\log \ad)$ and it requires $O(n \cdot \ad^3)$ words of global memory.

The high-level idea is to iteratively compress parts of the graph (without disconnecting it) such that the knowledge of the  maximum identifier of the compressed parts is always kept within the resulting graph. We repeat this process until there remains only one node, that knows \ID. Then, we backtrack the process by iteratively decompressing and broadcasting the knowledge about \ID. Eventually, we are left with the original graph where all nodes know \ID. 

As reasoned in \Cref{sec:intro} it is far from clear how to implement this simple outline with neither breaking the runtime nor the global memory bounds. From a high level point of view our algorithm consists of $O(1)$ \textbf{phases} and $O(1)$ \textbf{reversal phases}. During the phases, we first compress all light subtrees into single nodes (a procedure that we refer to as $\clst$) and then replace all paths by a single edge ($\cp$). In this section, we blackbox the properties of both procedures and prove  that $O(1)$ phases are sufficient to reduce the graph to a single node (\Cref{lem:lConstant}). By far the most technically involved part of our algorithm is the procedure $\clst$, which we explain in detail in \Cref{ssec:ccSinglePhaseTrees}. The phases are followed by reversal phases, in which we undo all compression steps of the regular phases in reverse order to spread \ID to the whole graph.

Let us be more formal and define the compression/decompression steps. Throughout the algorithm, every node $v$ keeps track of a variable $\id_v$, which is initially set to be the identifier of $v$. The intuition behind variable $\id_v$ is that it represents the largest identified $v$ has ``seen'' so far. Let us define compressing and decompressing operations for node $v$ and any node set $X$. Note that decompressing $X$ from $v$ is only defined for $X,v$ such that $X$ was at some point compressed into $v$.

\begin{itemize}
	\item Compress $X$ into $v$: set $\id_v \larr \max_u \{\id_u \mid u \in X\}$ remove $X$ (and its incident edges) from the graph. For any edge $\{x,y\}$ with  $x\in X$ and $v\neq y \notin X$ we introduce a new edge $\{v,y\}$. 
	
	\item Decompress $X$ from $v$: set $\id_u \larr \id_v, ~\forall u \in X$ and add $X$ (and its incident edges) back to the graph. Remove any edge $\{v,y\}$ that was added during the compression step of $X$ into $v$.
\end{itemize}

\paragraph{Phases.} We initialize $G_0$ as the input graph. From $G_0$, we derive a sequence $G_1, G_2, \dots, G_\ell$ of smaller trees until eventually, for some $\ell=O(1)$, it holds that $G_\ell=\{v\}$ for which $\id_v=\ID$. The tree $G_{i+1}$ ($0 < i \leq \ell$) is obtained from $G_{i}$ as follows: first compressing all light subtrees via $\clst(G_i,\ad)$ and call the resulting tree $G_i'$, then $G_{i+1}$ is the result of compressing all paths of $G_i'$ into single edges via $\cp(G'_i,\ad)$.

Throughout the sequence, we maintain the properties that compressions do not overlap, every $G_i$ is connected and non-empty, and that $\id_w=\ID$ for \emph{some} node $w \in G_i$.

\paragraph{Reversal Phases.} From $G_\ell=\{v\}$, we derive a reversal sequence $G_{\ell-1}, G_{\ell-2}, \dots, G_0$ such that any $G_i$ ($\ell > i \geq 0$) has the same node and edge sets as $G_i$ during the regular phases, and $\id_w=\ID$ for \emph{every} node $w \in G_i$. The tree $G_{i-1}$ is obtained from $G_{i}$ as follows: first decompressing all paths via $\dcp(G_i)$ and call the resulting tree $G_{i-1}'$, then $G_{i-1}$ is the result of decompressing all light subtrees via $\dclst(G'_{i-1})$. Note that in reversal phase $i$ we only decompress paths and subtrees that were compressed during the regular phase $i$.

\falgo{$\cc(G,\ad)$}{
	\item[] \hspace{-6mm} Initialize $G_0 \larr G$
	\item  For $i = 0,\dots,\ell-1$ phases:
	\begin{enumerate}
		\item $G'_i=\clst(G_i,\ad)$ \vspace{1mm} \\
		\textit{// If there are heavy nodes, all light nodes are compressed into the closest heavy node. Otherwise, all nodes are light and are compressed into a single node.}
		\item $G_{i+1}=\cp(G'_i,\ad)$ \vspace{1mm} \\
		\textit{// All paths are compressed into single edges.}
	\end{enumerate}
	\item For $i = \ell-1, \dots, 0$ reversal phases:
	\begin{enumerate}
		\item $G'_i=\dcp(G_{i+1})$ \vspace{1mm} \\
		\textit{// All paths that were compressed during Step 1(b) are decompressed.}
		\item $G_i=\dclst(G'_i)$ \vspace{1mm} \\
		\textit{// All light nodes that were compressed during Step 1(a) are decompressed from $v$.}
	\end{enumerate}
}

The correctness of \cc is contained in the following lemma.

\begin{lemma} \label{lem:lPhases} \label{lem:lReversalPhases}
There exists some $\ell$ such that
	\begin{enumerate}
\item 	after $\ell$ phases, graph $G_\ell$ consists of exactly one node $v$ for which $\id_v=\ID$.
\item 	after $\ell$ reversal phases, graph $G_0$ is the input graph and all nodes know \ID. 
\end{enumerate}
\end{lemma}

\begin{proof}
The proof is straightforward, given the thee essential lemmas (\Cref{lem:CompressLightSubTrees,lem:CompressPaths,lem:Decompress}) on the subroutines that we prove in the sections hereafter. Let us prove the two claims separately. 
\begin{enumerate} 
	\item Consider graph $G_i$ at the start of any phase $i$. We first claim that $G_i$ never becomes empty during phase $i$, for which there are two cases: either $G_i$ contains heavy nodes, or all nodes in $G_i$ are light. In the case of the former: in Step 1(a),  by \Cref{lem:CompressLightSubTrees}, if there are heavy nodes in the graph, they are never compressed. In Step 1(b), by \Cref{lem:CompressPaths}, all degree-2 nodes are compressed into single edges, leaving the graph non-empty. In the case of the latter, by \Cref{lem:CompressLightSubTrees}, we are left with a single node. Observe that since any tree always contains light nodes (leaves are always light), the number of nodes decreases in every phase, and the first part of the claim 1 holds for some $\ell$. Since $G_\ell=\{v\}$ is a result of consecutive compression steps applied to the input graph $G_0$ without disconnecting it, by the definition of compression, it holds that $\id_v=\ID$.
	
	\item \textbf{Observation.} \emph{Graph $G_i$ during reversal phases $i$ has the same node and edge sets as graph $G_i$ during phase $i$.}
	\begin{proof}
        We prove the claim by induction. The base case holds since $G_{\ell}$ from Step 1 is given directly to Step 2 as input. Assume that the claim holds for reversal phase $i+1$. By \Cref{lem:Decompress}, all nodes that were compressed in phase $i$ during Step 1(a) (resp. (b)) can decompress themselves in reversal phase $i$ during Step 2(b) (resp. (a)), proving the claim.
	\end{proof}
	
	Consider graph $G_\ell$ that consists of a single node $v$ for which $\id_v=\ID$ by \Cref{lem:lPhases}. Since graph $G_0$ after $\ell$ reversal phases (which is the input graph by the observation above) is a result of consecutive decompression steps applied to $G_\ell$,  by the definition of decompression, it holds that $\id_u=\ID$ for all $u\in G_0$. \qedhere
	\end{enumerate}
\end{proof}

\begin{restatable}[\clst]{lemma}{lemCompressLightSubTrees}
	\label{lem:CompressLightSubTrees}
	Let $G$ be a tree and $\ad \in [\diam(G),n^{\delta/8}]$.
	If $G$ contains a heavy node, then $\clst(G,\ad)$ returns a tree in which all light nodes of $G$ are compressed into the closest heavy node. If $G$ does not contain any heavy nodes, all nodes are compressed into a single node. The algorithm runs in $O(\log \ad)$ low-space \mpc rounds using $O(n \cdot \ad^3)$ words of global memory.
\end{restatable}
 
\begin{restatable}[\cp]{lemma}{lemCompressPaths}	\label{lem:CompressPaths} 
	For any tree $G$ and $\ad \in [\diam(G),n^{\delta/8}]$, $\cp(G,\ad)$ returns the graph that is obtained from $G$ by replacing all paths of $G$ with a single edge. The algorithm runs in $O(\log \ad)$ low-space \mpc rounds using $O(n \cdot \ad^2)$ words of global memory.
\end{restatable}

\begin{restatable}[\dcp,\dclst]{lemma}{lemDecompress}
	\label{lem:Decompress} All nodes that were compressed by \clst and \cp can be decompressed by \dclst and \dcp, respectfully. The algorithms run in $O(1)$ low-space \mpc rounds using $O(n)$ words of global memory.
\end{restatable}

We will now show that the number of phases of (and therefore reversal phases) is bounded by $O(1)$. In particular, we want to prove that after $\ell=O(1)$ phases, graph $G_\ell$ consists of exactly one node. After a clever observation in \Cref{lem:manyLightNodesPerHeavy}, we will prove the claim in \Cref{lem:lConstant}.
	
\begin{lemma} \label{lem:manyLightNodesPerHeavy}
	If $|G_{i+1}| \geq 2$, all nodes in $G_{i+1}$ were heavy in $G_i$. Moreover, for every leaf node $w \in G_{i+1}$ it holds that $\geq n^{\delta/8}$ light nodes were compressed into $w$ during phase $i$.
\end{lemma}

\begin{proof}
	Since $|G_{i+1}| \geq 2$ (and not $|G_{i+1}|=1$), by \Cref{lem:CompressLightSubTrees}, there must have been heavy nodes in $G_i$. Since all light nodes were compressed in phase $i$, all nodes in $G_{i+1}$ were heavy in $G_i$. Observe that even though $w$ is a leaf in phase $i+1$, it was not a leaf node in phase $i$, since leaf nodes are light by definition. Let $u$ be the unique neighbor of $w$ in $G_{i+1}$. We must show that $|G_{w \narr u}| > n^{\delta/8}$ and that $G_{w \narr u} \setminus w$ was compressed into $w$ during phase $i$. It must be that $|G_{w \narr u}| > n^{\delta/8}$, since otherwise, $w$ would have been light against $u$ in phase $i$. Nodes $G_{w \narr u} \setminus w$ were compressed into $w$ during phase $i$ by \Cref{lem:CompressLightSubTrees}, since $w$ was their closest heavy node (due to the graph being a tree).
\end{proof}

\begin{lemma} \label{lem:lConstant}
	After $\ell=O(1)$ phases, graph $G_\ell$ consists of exactly one node.
\end{lemma}

\begin{proof}
    Consider graph $G_i$ at the beginning of some phase $i$. If there are no heavy nodes in $G_i$, this is the last phase of the algorithm by \Cref{lem:CompressLightSubTrees}. If there is exactly one heavy node in $G_i$, we are also done by \Cref{lem:CompressLightSubTrees}. What remains to be proven is that if there are at least two heavy nodes in the graph, we reduce the size of the graph by a polynomial factor in $n$.
	
	Assume that there are at least 2 heavy nodes in graph $G_i$, and let us analyze what happens. In Step 1(a), all light nodes are compressed into the closest heavy node by \Cref{lem:CompressLightSubTrees}. In Step 1(b), all paths are compressed into single edges by \Cref{lem:CompressPaths}, leaving no degree-2 nodes in the graph (compressing paths never creates new degree-2 nodes). Consider graph $G_{i+1}$, which by \Cref{lem:CompressLightSubTrees} consists of the nodes that were heavy in $G_i$. By \Cref{lem:manyLightNodesPerHeavy} it also holds that during phase $i$, at least $n^{\delta/8}$ light nodes were compressed into every leaf node $w$ of graph $G_{i+1}$. It holds that
	\begin{align*}
		n_{i} \geq n_{i+1} + |\{w \in G_{i+1} \mid \deg_{G_{i+1}}(w)=1\}| \cdot  n^{\delta/8} 
		> n_{i+1} + n^{\delta/8} \cdot n_{i+1}/2 
		= n_{i+1} (1+n^{\delta/8}/2)
	\end{align*}
	and $n_{i+1} < n_i / (1+n^{\delta/8}/2) < 2n_i / n^{\delta/8}$~.
	
	The first strict inequality stems from the fact that there are no degree-2 nodes left after phase $i$, and hence the number of leaf nodes in $G_{i+1}$ is strictly larger that $n_{i+1}/2$. The proof is complete, as we have shown that if graph $G_{i}$ contains at least 2 heavy nodes, $G_{i+1}$ is smaller than $G_i$ by a factor of $\Theta(n^{\delta/8})$.
\end{proof}

The outline for the rest of this section is as follows. The procedure  $\clst$ and the proof of  \Cref{lem:CompressLightSubTrees} are presented in \Cref{ssec:ccSinglePhaseTrees}. This is the most technically involved part of our algorithm. The procedure $\cp$ and the proof of 
\Cref{lem:CompressPaths} are presented in \Cref{ssec:ccSinglePhasePaths}. The procedures $\dcp$ and $\dclst$  and the proof of \Cref{lem:Decompress} are presented in \Cref{ssec:ccReversalPhases}. In \Cref{sec:MPCdetails}, we show technical details how \cc can be implemented in the low-space \mpc model. 

\subsection{\maxid: Single Phase (\clst)}
\label{ssec:ccSinglePhaseTrees}

In this section, we focus on a single execution of $\clst(G,\ad)$ on a graph $G$ and prove \Cref{lem:CompressLightSubTrees}. With out loss of generality, we assume there are $n$ nodes in the graph---starting from the second phase of \cc we will actually use this algorithm on graphs with fewer than $n$ nodes.

At all times, every nodes $v$ has some set of nodes $S_v$ in its memory, which we initialize to $N(v)$. Set $S_v$ can be thought of as the node's view or knowledge. During the execution, $S_v$ grows, and if $|S_v| \geq 2n^{\delta/4}$, $v$ becomes \full. Similarly to definitions $G_{v \arr u}$ and $G_{v \narr u}$, let us define the following. For a node $v$ and a node $u \in N(v)$, let $S_{v \arr u}=S_v\cap G_{v\arr u}$. Also, let $S_{v \narr u} \coloneqq S_v \setminus S_{v \arr u}$. Recall the definition of $r_v(w)$: for every $w \in G$ let $r_v(w)$ be $u \in N(v)$ such that $w \in G_{v \arr u}$.

All nodes in the graph have the property that they are either light or heavy (see \Cref{def:lightheavy}). Initially, nodes themselves do not  know whether they are light or heavy, since these properties depend on the topology of the graph. During the algorithm each node is in one of the four states: \act, \happy, \full, or \sad. Initially, all nodes are \act. A node $v$ becomes \happy, if at some point during the execution, there exists $u \in N(v)$ such that such that $G_{v \narr u}\subseteq S_{v}$ and $|G_{v \narr u}| \leq n^{\delta/8}$. In that case, we say that node $v$ is \emph{\happy against} $u$. If a node, that is not \full, realizes that it can never become \happy (for example by having $|S_{v \arr u}|>n^{\delta/8}$ for two different neighbors $u$), it becomes \sad. Upon becoming \happy, \sad or \full, nodes do not partake in the algorithm except for answering queries from \act nodes. We call nodes \unhappy if they are in some other state than \happy (including state \act). The goal is that all light nodes eventually become \happy, and heavy nodes always remain \unhappy. 
When comparing the definitions of \happy and light, it is evident that when a node becomes \happy, it knows that it is light. Similarly, a node becoming \full or \sad knows that it is heavy.

For a node $v$ and any $X\subseteq N(v)$, define an exponentiation operation as
\begin{align*}
	\expo(X): ~~ S_v \larr \bigcup_{u \in X} \bigcup_{w \in S_{v \arr u}} S_{w \narr r_w(v)}. 
\end{align*}

We say that a node $v$ \emph{exponentiates towards (or in the direction of) $u$} if $u \in N(v)$ and $v$ performs $\expo(X)$ with $u\in X$.

\paragraph{High level overview of \clst.}

The algorithm consists of $O(\log \ad)$ iterations, in each of which nodes perform a carefully designed graph exponentiation procedure. The aim is for light nodes $v$ to become \happy by learning their subtrees $T_v$, after which, (certain) light nodes compress $T_v$ into their \unhappy neighbor.
If there were no memory constraints and every node could do a proper (uniform) exponentiation step in every iteration of the algorithm, i.e., execute $\expo(N(v))$, after $O(\log \ad)$ iterations \emph{all} nodes would learn the whole graph---a proper exponentiation step executed on all nodes halves the diameter---and the highest ID node could compress the whole graph into itself. However, uniform exponentiation would result in all nodes exceeding their local memory $O(n^\delta)$, and also significantly breaking the global memory requirement. Even if we were to steer the exponentiation procedure such that light nodes would learn a $D_{T_v}$ radius ball around them, where $D_{T_v}$ is the diameter of their light subtree, this would still break global memory. In fact, we cannot even do a single exponentiation step for all nodes in the graph without breaking memory bounds!
Hence, we need to steer the exponentiation with some even more stronger invariant in order to abide by the global memory constraint.

\begin{observation}
    If every light node $v$ keeps $O(|T_{v,u}|)$ nodes in its local memory for some (possibly unique) neighbor $u$ it is light against, this does not violate local memory $O(n^\delta)$ nor global memory $O(n \cdot \ad)$.
\end{observation}

\begin{proof}
    If there is a heavy node in the graph, $|T_{v,u}|$ is unique by \Cref{obs:heavyCCandTVunambiguous}. The claim follows by considering a hypothetical rooting of the tree at some heavy node and applying \Cref{lem:globalmemory}. Otherwise, the claim holds trivially because the graph is of size $\leq 2n^{\delta/8}$ by \Cref{cor:sizeOfGraph}.
\end{proof}
 
Inspired by the observation above, we aim to steer the exponentiation such that it is performed
in a balanced way, where a node learns roughly the same number of nodes in each direction (or sees only leaves in one direction). In fact, we do not want to exponentiate in a direction if that exponentiation step would provide us with $\gg |T_v|$ nodes. This step is further complicated as nodes neither know whether they are heavy or light nor do they know the size of their subtree. In our algorithm that is presented below we perform a careful probing for the number of nodes into all directions to determine in which directions we can safely exponentiate without using too much memory. In the probing procedure $\pro$, a node $v$ computes 
$B_{v \arr u} = \sum_{w \in S_{v \arr u}} |S_{w \narr r_w(v)}|$  for every neighbor $u \in N(v)$ as an estimate for the number of nodes it may learn when exponentiating towards $u$. This estimate may be very inaccurate and may contain a lot of doublecounting. 
In \Cref{sssec:probing}, we present the full procedure and prove the following lemma. 
\begin{restatable}[\pro]{lemma}{lemMainProbing}
	\label{lem:mainProbing}
	Consider an arbitrary iteration of algorithm $\clst$. Then algorithm $\pro(\ad)$ returns:
	\begin{enumerate}
		\item[(i)] $\fd \subseteq N(v)$ such that if we were to exponentiate in all directions, we would obtain $|S_{v\rightarrow u'}| > n^{\delta/8}$ for all $u' \in \fd$ and $|S_{v \arr u'}| \leq n^{\delta/8} \cdot \ad$ for all $u' \in N(v) \setminus \fd$.  
		
		\item[(ii)] $\ld \in N(v)$ (returned if $\fd =\emptyset$) such that if we were to exponentiate in all directions, we would obtain $|S_{v\rightarrow \ld}|\geq |S_{v\rightarrow u'}|$ for all $u'\in N(v)$ and $|S_{v \arr \ld}| \leq n^{\delta/8} \cdot \ad$. 

	\end{enumerate}
	\pro can be implemented in $O(1)$ low-space \mpc rounds, using $O(n\cdot \ad^3)$ global memory. It does not alter the state of $S_v$ for any node $v$ in the execution of $\clst$. 
\end{restatable}

The main difficulty of \clst lies in ensuring the global and local memory constraints (\Cref{lem:rootingGlobalMemory,lem:rootingLocalMemory}) that prevent us from blindly exponentiating in all directions, while at the same time ensuring enough progress for light nodes such that every light node becomes \happy by the end of the algorithm (\Cref{lem:correctnessLemmaLight}).

\falgo{$\clst(G_i,\ad)$}{
	\item[] \hspace{-6mm} All nodes are active. Initialize $S_v \larr N(v)$. If $|S_v| > n^{\delta/8}+1$, $v$ becomes sad. 
	\item For $O(\log \ad)$ iterations:
	\begin{enumerate}
		\item $\fd, \ld  \larr \pro(\ad)$ \vspace{1mm} \\ // The properties of $\pro(\ad)$ are formally stated in \Cref{lem:mainProbing}. Informally, $\fd \subseteq N(v)$ contains directions with $>n^{\delta/8}$ nodes, and  $\ld$ contains the direction with the largest number of nodes if $\fd = \emptyset$.
		
		\item If $|\fd| \geq 2$, $v$ becomes sad.
		
		\item If $|\fd| = 1$:
		
		\begin{enumerate}
			\item Perform $\expo(N(v) \setminus \fd)$
		\end{enumerate}
		
		\item If $|\fd| = 0$:
		
		\begin{enumerate}
			\item Perform $\expo(N(v) \setminus \ld)$
			\item If $v$ is in $S_w$ for some $w$, add $w$ to $S_v$ \hfill \text{// ensure symmetric view}
		\end{enumerate}
		
		\item Node $v$ asks nodes $w \in S_v$ whether or not they are happy against $r_w(v)$, and if so, what is the size of subtree $T_{w,r_w(v)}$. Node $v$ can locally compute if it can become happy by learning subtrees $T_{w,r_w(v)}$. If $v$ can, it asks for them and becomes happy. 
	\end{enumerate}
	// After Step 1, all light nodes are happy, and all heavy nodes are unhappy (\Cref{lem:correctnessLemmaLight})
	\item Happy nodes $v$ with an unhappy neighbor $u$ compress $S_{v \narr u}=G_{v\narr u}$ into $u$. \vspace{1mm} 
	
	\item  Nodes $v$ that are happy against $u$ such that $u$ is happy against $v$ update $S_v \larr S_v \cup S_u$ and compress $S_v$ into the highest ID node in $S_v$.  \vspace{1mm} 
}

In \Cref{sssec:progressAndCorrectness}, we discuss the measure of progress and correctness, with the final correctness proof of \Cref{lem:CompressLightSubTrees}. In \Cref{sssec:memoryBounds}, we discuss local and global memory bounds, with the final memory proofs of \Cref{lem:CompressLightSubTrees}. The \mpc implementation is deferred to \Cref{sec:MPCdetails}.

\subsubsection{Measure of Progress and Correctness} \label{sssec:progressAndCorrectness}
We begin by proving the measure of progress and correctness, which will give us the means to analyze the memory requirements as if the tree was rooted. 

\begin{lemma} \label{lem:lightExpo}
	Let $v$ be a node that is light against neighbor $u$. If in some iteration of $\clst$, $v$ exponentiates in the direction of $u$, i.e., it performs $\expo(X)$ with $u\in X$, the size of the resulting set $S_{v \arr u}$ is bounded by $|T_v|$.
\end{lemma}

\begin{proof}
	Consider an arbitrary iteration of the algorithm. 
	If $u\in \fd$, we do not exponentiate towards $u$, so there is nothing to prove. If $u\notin \fd$, but $\fd \neq \emptyset$, there is some $w \neq u$ such that, by the Probing \Cref{lem:mainProbing}, $|G_{v\arr w}| > n^{\delta/8}$, which is a contradiction to $v$ being light against $u$.
	
	Hence, consider the case that  $\fd=\emptyset$. If $\ld=u$, we do not exponentiate towards $u$ and there is nothing to prove. If $\ld\neq u$, then we exponentiate towards $u$ and by \Cref{lem:mainProbing} $(ii)$, we have $|S_{v \arr u}| \leq |G_{v \arr \ld}| \leq |T_v|$.
	\end{proof}

\begin{lemma} \label{lem:lightNeverFullSad}
	In any iteration of \clst, a light node neither becomes \full nor \sad. 
\end{lemma}

\begin{proof}
	Node $v$ never becomes \full due to initialization $S_v \larr N(v)$, since for a light node it must hold that $|N(v)| \leq |T_v|+1 \leq n^{\delta/8}+1 < 2n^{\delta/4}$. During execution, $S_v$ grows only in Steps 1(c)--(e). During (c), it must be that $\fd=u$, since otherwise it would imply that $|T_v|>n^{\delta/8}$. Hence, as a result of (c), $v$ cannot become \full. During (d)i, if $\expo(X)$ with $u \not\in X$, it holds that $X \subset T_v$ and $v$ cannot become \full. Otherwise if $u \in X$, by \Cref{lem:lightExpo}, $v$ cannot become \full. Node $v$ cannot become \full even when performing Step 1(d)ii, since a hypothetical exponentiation step in the direction of $u$ would yield a set that is bounded by $n^{\delta/8} \cdot \ad < 2n^{\delta/4}$ (\fd is empty and $u$ is \ld). During (e), node $v$ becomes happy against $u$ and hence $|S_{v \narr u}| \leq n^{\delta/8}$. In the worst case, $|S_{v \arr u}| < n^{\delta/8} \cdot \ad \leq n^{\delta/4}$. Hence, as a result of (e), $v$ cannot become full.
	
	A node can become sad only if its degree is too large, or in Step 1(b). A light node $v$ never becomes sad since it must hold that $|N(v)| \leq |T_v|+1 \leq n^{\delta/8}+1$, and $v$ cannot have two or more neighbors $u$ with $G_{v \arr u}>n^{\delta/8}$ (one neighbor would have to be in $T_v$, implying that $|T_v|>n^{\delta/8}$).
\end{proof}

For the proofs of the next two lemmas, let $G=(V,E)$ be the input graph, and consider graph $G'=(V',E')$ such that $V'=V$ and $E' = E \cup \{ \text{ $\{v,w\} \mid v,w\in V \text{and } w \in S_v$ or $v \in S_w$} \}$. 

\begin{lemma}[Measure of progress] \label{lem:pathShortening}
	At the start of any iteration $j$, consider a light (but still \act) node $v$, and the longest shortest path $P^j_{vw}$ in $G'$ between $v$ and an a leaf node $w \in T_v$. If $|P^{j}_{vw}| \geq 4$ holds, then holds that $|P^{j+1}_{vw}| \leq \lceil 3/4 \cdot |P^j_{vw}| \rceil$ holds. 
\end{lemma}

\begin{proof}
	Consider any subpath $P_{x_1x_5} = \{x_1,x_2,x_3,x_4,x_5\} \subseteq P^j_{vw}$ of length $4$. For $1\leq i\leq 5$, let  $S_{x_i}$ ($S'_{x_i}$) be the memory of node $i$ at the start (end) of iteration $j$. 
    Note that all nodes on the path are light. By \Cref{lem:lightNeverFullSad}, nodes in $P_{x_1x_5}$ never get \full nor \sad, and hence always exponentiate in all but one direction (either $\fd$ or $\ld$). 

\smallskip

\noindent	\textbf{Claim. } For $1 \leq i < 5$ it holds that $x_{i+1} \in S_{x_i}$.

\begin{proof}
	Since edge $\{x_i,x_{i+1}\}$ exists in $G'$, it must be either that either $x_{i+1} \in S_{x_i}$ or $x_{i} \in S_{x_{i+1}}$. In the first case the claim holds, so consider the latter. Since $x_i$ is light, $r_{x_i}(x_{i+1})$ has never been in \fd for $x_i$. Hence, whenever $x_{i+1}$ had added $x_i$ to $S_{x_{i+1}}$, either $x_i$ added $x_{i+1}$ to $S_{x_{i}}$ via exponentiation, or via Step 1(d)ii.
\end{proof}

We continue with proving that the path shortens. It is sufficient to prove that for some $i,j \in [1,5]$, $i \neq j$, it holds that $x_i \in S'_j \setminus S_j$, as this shortens the path between $x_1$ and $x_5$ by one edge. 
	\begin{enumerate}
	
		\item If $x_2 \not\in S_{x_3}$: Since there is an edge in $G'$ such that $x_2 \not\in S_{x_3}$, it means that $x_2$ added $x_3$ to $S_{x_2}$ during some iteration $j$, and since $x_3$ did not add $x_2$ in Step 1(d)ii of iteration $j$, direction $r_{x_3}(x_2)$ must have been in \fd for $x_3$. Hence, $x_3$ will exponentiate in all directions besides $r_{x_3}(x_2)$. In particular, as $r_{x_3}(x_4)\neq r_{x_3}(x_2)$, $x_3$ will exponentiate towards $x_4$. As $x_5\in S_{x_4}$, we obtain $x_5\in S'_{x_3}\setminus S_{x_3}$. This creates an edge between $x_3$ and $x_5$ in $G'$ and shortens the path from $4$ to $3$, i.e., by a factor $3/4$.
		\item If $x_2 \in S_{x_3}$: Assume that $x_1\in S_{x_2}$. Since nodes in $P_{x_1 x_5}$ exponentiate in all but one direction, node $x_3$ will exponentiate either towards $x_2$ or $x_4$ (it must be that $x_4 \in S_{x_3}$ by the claim above). If $x_3$ exponentiates towards $x_4$, we obtain $x_5\in S'_{x_3}\setminus S_{x_3}$ as $x_5\in S_{x_4}$. If $x_3$ exponentiates towards $x_4$, we obtain $x_1\in S'_{x_3}\setminus S_{x_3}$ as $x_1\in S_{x_2}$. If $x_1 \not\in S_{x_2}$, we can apply the analysis of 1. for node $x_2$. 
		\qedhere
	\end{enumerate}
\end{proof}

\lemCorrectnessLemmaLight

\begin{proof}
	Let us adopt the notation of the proof of \Cref{lem:pathShortening}. Since \Cref{lem:pathShortening} holds for any light node $v$, after $j=O(\log \ad)$ iterations it must holds that $|P^j_{vw}| \leq 3$ because $\ad \in [\diam(G),n^{\delta/8}]$. Let the resulting path be $P=\{x_1,x_2,x_3,x_4\}$, where $x_4$ is a leaf node. It must be the case that if $x_2$ learns $S_{x_3 \narr x_2}$ for all possible nodes $x_3$, node $x_2$ becomes happy. Hence, in Step 1(e), $P$ shortens by one. Eventually, after two iterations, $P$ is of length one, and $x_1$ becomes happy.
	
For the second part of the claim it is sufficient to show that heavy nodes never become happy. Recall that heavy nodes are defined as nodes that are not light. Hence, for a heavy node $v$, there does not exist a neighbor $u \in N(v)$ such that $|G_{v \narr u}| \leq n^{\delta/8}$. This implies that during the algorithm, it is not possible for $|S_{v \narr u}| = |G_{v \narr u}| \leq n^{\delta/8}$ for any $u \in N(v)$. Hence, heavy nodes never become happy.
\end{proof}

\begin{proof}[Proof of \Cref{lem:CompressLightSubTrees} (Correctness)]
By \Cref{lem:correctnessLemmaLight}, we know that after $O(\log\ad)$ iterations all light nodes of $G$ become \happy, while all heavy nodes always remain \unhappy. In order to prove the correctness of \Cref{lem:CompressLightSubTrees}, we need to show that all light trees are compressed into the closest heavy node, if a heavy node exists, and that the whole tree is compressed into a single node if there are no heavy nodes. We consider both cases separately. Also consult \Cref{fig:exampleGraphs} for an illustration of both cases.

\textbf{Case 1 (there are heavy nodes).}
Consider a light node $v$. As there are heavy nodes, \Cref{obs:heavyCCandTVunambiguous} implies that there is a unique neighbor $u\in N(u)$ against which $v$ is light. Let $T_v=G_{v\narr u}$.  Now, by \Cref{lem:correctnessLemmaLight}, $v$ is happy at the end of the algorithm, i.e.,  there is a neighbor $u'\in N(v)$ for which $S_{v \narr u'} = G_{v \narr u'}$ and $|G_{v \narr u'}| \leq n^{\delta/8}$. The latter condition says that $v$ is light against $u'$ and due to the earlier discussion we deduce that $u=u'$ and $T_v=G_{v \narr u}=S_{v \narr u}\subseteq S_v$ holds.
In summary, for every light node $v$, the tree $T_v$ (that does not depend on the algorithm) is contained in $S_v$. By \Cref{obs:heavyCCandTVunambiguous}, heavy nodes induce a single connected component, and hence every light node $v$ is contained in the subtree of some light node $v'$ that has a heavy neighbor $u'$. Since we are in a tree, $u'$ is the closest heavy node for $v'$, and in particular, for all light nodes $v \in T_v'$. By  \Cref{lem:correctnessLemmaLight}, $v'$ is happy at the end of the algorithm, and $u'$ remains unhappy. Performing Step 2 fulfills the first claim of \Cref{lem:CompressLightSubTrees}. Step 3 is never performed, since there are no happy nodes left in the graph.

\textbf{Case 2 (all nodes of $G$ are light).} Step 2 of $\clst$ is never performed, since all (light) nodes are happy due to \Cref{lem:correctnessLemmaLight}.  For the sake of analysis, let each node $v$ put one token on each incident edge $\{v,u\}$ for which $G_{v\narr u}\subseteq S_v$ holds.  As all (light) nodes are happy, i.e., there is a neighbor $u$ such that $G_{v\narr u}\subseteq S_v$ holds, the total number of tokens is at least as large as the number of nodes. Since the graph is a tree, at least one edge receives two tokens. Let $\{u,v\}$ be such an edge and observe that $G_{v\narr u}\subseteq S_v$ and $G_{u\narr v}\subseteq S_u$. Due to \Cref{lem:newLemmaThreeFour}, $G_{v\narr u}\cup G_{u\narr v}=G$ and after Step 3 of \clst both nodes have the complete tree in their memory and both nodes trigger a compression of the whole tree into the largest ID node. 
  
The edge $\{u,v\}$ with the above properties is not unique, but after Step 3, the endpoints of any edge having these properties yield the exact same compression.
\end{proof}

\subsubsection{Local and Global Memory Bounds} \label{sssec:memoryBounds}

The most difficult part is proving the memory bounds when there are heavy nodes. If there were no memory limitation, \Cref{lem:lPhases,lem:lConstant} (building up on versions of \Cref{lem:CompressLightSubTrees,lem:CompressPaths,lem:Decompress} without memory limitations) already imply that after $O(1)$ phases of \cc, there is exactly one node left in the graph. Denote this node by $r$. Node $r$ has never been compressed by definition. For the sake of analysis, we assume a rooting of $G$ at $r$. We emphasize that fixing a rooting is only for analysis sake, and we do not assume that the tree is actually rooted beforehand. We define $T(v,r)$ as the subtree rooted at $v$ (including $v$ itself), as if tree $G$ was rooted at $r$. 

\begin{observation}
    Consider tree $G$ with at least one heavy node during an arbitrary iteration of \clst. For every light node $v$, it holds that $T(v,r)=T_v$, and every heavy node $u$ has a unique subtree $T(u,r)$. 
\end{observation}

Recall that the definition of $T_v$ for a light node $v$ was independent from any algorithmic treatment. Still, it holds that $T_v$ equals $T(v,r)$ (that depends on our algorithm as the node $r$ depends on it).
The next lemma states that a node only exponentiates into the direction of root $r$ if it is safe to do so in terms of memory constraints. In spirit, it is very similar to \Cref{lem:lightExpo}, with the slight difference that it applies to all nodes, and we prove the claim using a hypothetical rooting of the tree.

\begin{lemma} \label{lem:memoryDirectionExponentiation} 
    Let $v$ be any node with a parent $u$ (according to the hypothetical rooting at $r$). If in some iteration of \clst when there are heavy nodes, node $v$ exponentiates in the direction of $u$, i.e., it performs $\expo(X)$ with $u \in X$, the size of the resulting set $S_{v \arr u}$ is bounded by $|T(v,r)| \cdot \ad$.
\end{lemma}

\begin{proof}
    If $v$ performs $\expo(X)$ such that $u \in X$, there must exists either $w = \fd$ or $w = \ld$ such that $w \in T(v,r)$. If $w = \fd$, it implies that $v$ is heavy, and by \Cref{lem:mainProbing} $(i)$, we have that $|S_{v \arr u}| \leq n^{\delta/8} \cdot \ad < |T(v,r)| \cdot \ad$. If $w = \ld$, by \Cref{lem:mainProbing} $(ii)$, we have that $|S_{v \arr u}| \leq |G_{v \arr w}| \leq |T(v,r)|$.
\end{proof}

\begin{observation} \label{obs:EXPduplicates}
    When node $v$ performs an exponentiation step, multiple nodes $w$ can send the same node to $v$, resulting in duplicates in set $S_v$. After every exponentiation step, node $v$ has to locally remove these duplicates. As a result, when bounding the memory of a node, we have to take into account the momentary spike in global memory due to duplicates. This momentary spike can at most result in an $\ad$ factor increase in the memory bounds. 
\end{observation}
\begin{proof}
     When node $v$ exponentiates, nodes $w$ send $S_{w \narr r_w(v)}$ and not $S_w$. Consider node $x$ that $v$ has received via exponentiation, and consider the unique path $P_{vx}$ between $v$ and $x$. Since only nodes $w\in P_{vx}$ could have sent $x$ to $v$, and $|P_{vx}|\leq D \leq \ad$, $x$ has at most $\ad$ duplicates in $S_v$.
\end{proof}

\begin{lemma} \label{lem:rootingGlobalMemory}
	In \clst, the global memory never exceeds $O(n \cdot \ad^3)$.
\end{lemma}

\begin{proof}

    The global memory $O(n \cdot \ad^3)$ of \pro (Step 1(a)) follows from \Cref{lem:mainProbing}. Hence, we analyze the global memory excluding Step 1(a).
    
    When all nodes are light, by \Cref{cor:sizeOfGraph}, the size of the graph is $\leq 2n^{\delta/8}$. When taking duplicates into account (\Cref{obs:EXPduplicates}), since $\ad\leq n^{\delta/8}$, even if the whole graph is in the local memory of every node, this does not violate global memory constraints. For the rest of the proof assume that there is at least one heavy.   
    
    Consider an arbitrary iteration $j$ of the algorithm when there are heavy nodes. Define set $C_v \subseteq S_v$ as the set of nodes that $v$ has added to $S_v$ as a result of performing $\expo$ in all iterations up to iteration $j$. Let $u$ be the parent of $v$ (according to the hypothetical rooting at $r$). For that $u$, define $C_{v \arr u} \coloneqq C_v\cap G_{v\arr u}$. For a node $v$, we have
	\begin{align*}
		|C_v| &\leq |C_{v \arr u}| + \sum_{w \in N(v) \setminus u} |S_{v \arr w}| \leq |T(v,r)|\cdot \ad + |T(v,r)| = (1 + \ad) |T(v,r)|
	\end{align*}

	The bound on $|C_{v \arr u}|$ is obtained by applying \Cref{lem:memoryDirectionExponentiation} for the last iteration where $v$ has exponentiated in the direction of $u$, and the bound on the sum is by the definition of $T(v,r)$. Observe the crucial difference between $S_v$ and $C_v$. Set $S_v$ may contain nodes that are \emph{not} a result of $v$ performing $\expo$, but rather the result of Step 1(d)ii of the algorithm, where some other node $w$ has added $v$ to $S_w$. However, this can result in at most a factor-2 overcounting for every node. Combining this with the duplicates of \Cref{obs:EXPduplicates} results in global memory
	\begin{align*}
	   \ad \cdot \sum_{v \in V} |S_v| = \ad \cdot \sum_{v \in V} 2 |C_v| \leq \ad \cdot \sum_{v \in V} 2 (1+\ad) |T(v,r)| = O(n \cdot \ad^3),
	\end{align*}
	where the bound on $\sum_{v \in V} |T(v,r)|$ is due to \Cref{lem:globalmemory}.
\end{proof}

\begin{lemma} \label{lem:rootingLocalMemory}
	In \clst, the local memory of a node $v$ never exceeds $O(n^\delta)$.
\end{lemma}

\begin{proof}
    The local memory of \pro  (Step 1(a)) follows from \Cref{lem:mainProbing}. Hence, we analyze the local memory excluding Step 1(a).

	   When all nodes are light, by \Cref{cor:sizeOfGraph}, the size of the graph is $\leq 2n^{\delta/8}$. When taking duplicates into account (\Cref{obs:EXPduplicates}), since $\ad\leq n^{\delta/8}$, even if the whole graph is in the local memory of every node, this does not violate global memory constraints. For the rest of the proof assume that there is at least one heavy.
	   
	   Consider the start of an arbitrary phase $i$. If $\deg(v) > n^{\delta}$, we defer the discussion to \Cref{lem:MPCdetails} on the \mpc implementation details. Assuming $\deg(v) \leq n^{\delta}$, we prove the claim by induction. During the algorithm, the size of the local memory is at most of order $|S_v| \cdot \ad$ (the extra $\ad$ factor is due to \Cref{obs:EXPduplicates}). The claim clearly holds in the first iteration when $S_v$ is initialized as $N(v)$. Observe that if $\deg(v) > n^{\delta/8}+1$, node $v$ becomes $\mathsf{sad}$. Hence, we can further assume that $\deg(v) \leq n^{\delta/8}+1$. Assume the claim holds in iteration $j$. We perform a case distinction on the different changes of $S_v$, and show that for a node $v$, it holds that $|S_v|=O(n^{7\delta/8})$, implying that $|S_v| \cdot \ad = O(n^\delta)$ since $\ad \leq n^{\delta/8}$.
	\begin{itemize}
		\item Step 1(c) and $|\fd|=1$, \vspace{1mm} \\ $|S_v|$ becomes at most $|S_{v \arr \fd}| + (\deg(v)-1) \cdot n^{\delta/8} \cdot \ad \leq 2n^{\delta/4} + (n^{\delta/8})^3 < n^{7\delta/8}$. The term $|S_{v \arr \fd}|$ has a (loose) upper bound of $2n^{\delta/2}$, since $v$ is not full. Observe that exponentiating in all directions except $\fd$ yields $\leq n^{\delta/8} \cdot \ad$ nodes per direction by \Cref{lem:mainProbing} $(i)$, and that $\ad \leq n^{\delta/8}$ by assumption.
		\item Step 1(d) and $|\fd|=0$, \vspace{1mm} \\ $|S_v|$ becomes at most $\deg(v) \cdot n^{\delta/8} \cdot \ad \leq (n^{\delta/8}+1) \cdot (n^{\delta/8})^2 < n^{7\delta/8}$. Observe that exponentiating in any direction yields $\leq n^{\delta/8} \cdot \ad$ nodes per direction by \Cref{lem:mainProbing} $(ii)$ ($\fd$ is empty), and that $\ad \leq n^{\delta/8}$ by assumption.
		\item Step 1(e), \vspace{1mm}
		\\ If a node $v$ becomes happy against $u$, $|S_v|$ becomes $|T(v,r)|+|S_{v \arr u}| \leq n^{\delta/8} + 2n^{\delta/4} < n^{7\delta/8}$, where $n^{\delta/8}$ is an upper bound for $|T(v,r)|$ since it is light, and $2n^{\delta/4}$ is (loose) upper bound on $|S_{v \arr u}|$ since $v$ is not full.
	\end{itemize}
	
	Hence, the claim holds in iteration $j+1$.
\end{proof}

\begin{proof}[Proof of \Cref{lem:CompressLightSubTrees} (Memory bounds)]
The local memory bounds follow from \Cref{lem:rootingLocalMemory}, and the global memory bounds follow from   \Cref{lem:rootingGlobalMemory}.
\end{proof}

\subsubsection{Probing} \label{sssec:probing}

Our probing procedure is an integral part of \clst, as it steers the exponentiation of every node such that, informally, a node never learns a (significantly) larger neighborhood in the direction of the root (which is imagined only for the analysis), than in the direction of its subtree.

\lemMainProbing*

\falgo{$\pro(\ad)$}{
	\item For every neighbor $u \in N(v)$, compute 
	$B_{v \arr u} = \sum_{w \in S_{v \arr u}} |S_{w \narr r_w(v)}|$.
	
	\item Define $\fd \coloneqq \{ u \in N(v) \mid B_{v \arr u} \geq n^{\delta/8} \cdot \ad \}$.
	
	\item If $|\fd|>0$, define $\ld \coloneqq \emptyset$. Otherwise, let $u_{\max} = \argmax_{u \in N(v)} \{B_{v \arr u}\}$ and if $B_{v \arr u_{\max}} \geq \ad \cdot B_{v \arr u'}$ for all $u' \in N(v) \setminus u_{\max}$
	\begin{enumerate}
		\item define $\ld \coloneqq u_{\max}$,
		\item otherwise,  perform $\expo(N(v))$ and define $\ld \coloneqq  \argmax_{u \in N(v)} \{ |S_{v \arr u}|\}$.
	\end{enumerate}

	\item Return $\fd, \ld$.
}

\begin{lemma} \label{lem:DApprox}
	If a node $v$ were to perform $\expo(u)$ for a neighbor $u \in N(v)$, it would hold that $B_{v \arr u} / \ad \leq |S_{v \arr u}| \leq B_{v \arr u}$.
\end{lemma}

\begin{proof}
	Recall the definition of $B_{v \arr u}$. Let us compute how many times a node $w$ in $B_{v \arr u}$ can be overcounted. Consider the unique path $P_{vw}$ from $v$ to a node $w$. Observe that out of the nodes in $S_{v \arr u}$, node $w$ is in  $S_{x \narr v}$ only for nodes $x \in P_{vw}$. Since $|P_{vw}| \leq D + 1 \leq \ad+1$, any node $w$ is overcounted at most $\ad$ times, completing the proof.
\end{proof}

\begin{proof}[Proof of \Cref{lem:mainProbing}]
	Combining the condition $|B_{v \arr u}| \geq n^{\delta/8} \cdot \ad$ of Step 2 and the $\ad$-factor overcounting of \Cref{lem:DApprox} proves the properties of $\fd$. The properties of $\ld$ hold by definition: in the case of Step 3(a), $u_{\max}$ is the largest direction by \Cref{lem:DApprox}, and in the case of Step 3(b), we exponentiate and find the absolute values. Local memory is respected in Step 1, since node $v$ only aggregates an integer from every other node in $S_v$. More importantly, it is respected in Step 3: a node performing $\pro$ has $\deg(v)<n^{\delta/8}+1$ (otherwise it is sad), $\expo(N(v))$ is only performed if all directions yield $\leq n^{\delta/8} \cdot \ad$ nodes (\fd is empty), and we are promised that $\ad \leq n^{\delta/8}$.
	
	Global memory is respected by a clever observation similar to \Cref{lem:rootingGlobalMemory}. Similarly to \Cref{sssec:memoryBounds}, assume we have a rooting at some node $r$, and that node $u$ is the parent of $v$. We want to bound the size of the resulting set $S_{v \arr u}$ if node $v$ performs $\expo(N(v))$. In particular, we want to show that  $|S_{v \arr u}| \leq |S_{v \arr w}| \cdot \ad^2 \leq |T(v,r)| \cdot \ad^2$ for \emph{some} node $w \in N(v) \setminus u$. Towards contradiction, assume that $|S_{v \arr u}| > |S_{v \arr w}| \cdot \ad^2$ for all $w \in N(v) \setminus u$. It must then hold that
	\begin{align*}
	   B_{v \arr u} \geq |S_{v \arr u}| > |S_{v \arr w}| \cdot \ad^2 \geq B_{v \arr w}/\ad \cdot \ad^2 = B_{v \arr w} \cdot \ad
	\end{align*}
	
	for all $w \in N(v) \setminus u$ by \Cref{lem:DApprox}. However, this implies that $u$ would have been chosen as $u_{\max}$, \ld would have been defined as $u_{\max}$, and $\expo(N(v))$ would have never been performed; we have arrived at a contradiction. It holds that $|S_{v \arr u}| \leq |T(v,r)| \cdot \ad^2$, which bounds set $S_v$ of every node by $(\ad^2 + 1) \cdot |T(v,r)|$, and by \Cref{lem:globalmemory}, the global memory is bounded by
	\begin{align*}
	    \sum_{v \in V} |S_v| \leq (\ad^2 +1) \sum_{v \in V} |T(v,r)| \leq (\ad^2 +1) (D+1) \cdot n = O(n \cdot \ad^3).
	\end{align*}

	Regarding \mpc implementation, $\pro$ only performs $\expo(N(v))$ (implementability proven in \Cref{lem:MPCdetails}) and computes $B_{v \arr u}$, which is only a modified version of $\expo(N(v))$: instead of nodes $w$ sending $S_{w \narr r_w(v)}$ to node $v$, they only send $|S_{w \narr r_w(v)}|$.
\end{proof}

\subsection{\maxid: Single Phase (\cp)} \label{ssec:ccSinglePhasePaths}

Let us prove the following lemma, which allows us to compress all paths in the tree into single edges. This operation does not create new paths or disconnect the graph.

\lemCompressPaths*

We describe a algorithm, which we denote as \cp, and which we run on every path $P \subseteq G$. A path only includes consecutive degree-2 nodes. Similarly to \clst, every node $v \in P$ has some set $S_v$ in its memory, which we initialize to $N_P(v)$. Every node performs $\expo(N(v))$ until $S_v$ no longer grows, whereupon, for every node $v$, it holds that $S_v=P$. The highest ID node $w \in P$ figures out the endpoints $x,y$ of path $P$ in $G$ (which either have degree 1 or $\geq 3$). Then, w.l.o.g., assume that $\text{ID}(x)>\text{ID}(y)$, whereupon $w$ compresses $P$ into $x$. By the definition of compression, node $w$ also creates edge $\{x,y\}$. 

\begin{proof}[Proof of \Cref{lem:CompressPaths}]
	After performing \cp, every node $v\in P$ learns path $P$, i.e., it holds that $S_v=P$, after $O(\log \ad)$ rounds, since the path is of length at most $\diam(G)$ and $\ad \in [\diam(G),n^{\delta/8}]$. Node $w$ can learn $x,y$ by asking for the neighbors (that are in $G \setminus P$) of the leaf.
	
    Because $|P| \leq D \leq \ad \leq n^{\delta/8}$, the local memory of a node is bounded by $n^{\delta/4}$ (when taking \Cref{obs:EXPduplicates} into account). The global memory is respected since in the worst case, all nodes have at most $\ad^2$ nodes in memory (when taking \Cref{obs:EXPduplicates} into account). Compressing and creating a new edge $\{x,y\}$ comprises of sending a constant sized message to both $x$ and $y$. Even in the case when $x$ or $y$ are endpoints to multiple paths, their total incoming message sizes are $O(\deg(x))$ and $O(\deg(y))$. The small caveat to this scheme is that if $\deg(x)$ or $\deg(y)$ are $>n^{\delta/8}$, we have to employ the aggregation tree structure as discussed in \Cref{lem:MPCdetails}. The implementation details of performing \expo can also be found in \Cref{lem:MPCdetails}.
\end{proof}

\subsection{\maxid: Single Reversal Phase} \label{ssec:ccReversalPhases}

A single reversal phase consist of steps \dcp and \dclst. In the former, we essentially reverse \cp, and in the latter, we reverse \clst. We prove the following.

\lemDecompress*

Let us introduce both steps formally.

\begin{itemize}
	\item \dcp. For every path $P$ that was compressed in phase $i$ into node $x$, node $x$ decompresses $P$ from itself. 

	\item \dclst. Every node $v$ that had compressed $T_v$ into a neighbor $u$ (or itself), decompresses $T_v$ from $u$ (or itself).
\end{itemize}

\begin{proof}[Proof of \Cref{lem:Decompress}]
	As long as the nodes $X$ that a node $v$ wants to decompress are in its local memory, both steps are clearly correct and implementable in $O(1)$ low-space \mpc steps. Observe that all nodes $v$ that decompress a node set $X$, have at some point compressed set $X$ and hence, have had $X$ in local memory (in the form of $S_v$). By simply retaining set $X$ in memory until it is time to decompress, we fulfill the requirement.
\end{proof}

\section{Connected Components (\coco)} \label{sec:CC}
\label{sec:ConnectedComponentsTheReal}
By \Cref{lem:maxid}, we can solve \maxid on any tree in $O(\log \ad)$ time using \cc. The algorithm requires $O(m \cdot \ad^3)$ words of global memory and value $\ad \in [\diam(G),n^{\delta/8}]$ as input. This section is mostly devoted to showing how to use \cc to solve the connected components (\coco) problem. 

\begin{definition}[The \coco Problem] \label{def:CCproblem}
Given a graph with unique identifiers for each node, and disconnected components $C_1,\dots,C_k$, every node $v \in C_i$ outputs the maximum identifier of $C_i$.
\end{definition}

Observe that \cc actually solves \coco for the case when the input graph is a single tree. We show how to extend \cc to solve \coco for forests, effectively proving the upper bounds of the following theorem.

\thmCC*

The proof is contained in \Cref{ssec:CCproof} with references to subroutines from \Cref{ssec:ccForestsAndStability,ssec:ccDiameter,ssec:ccProcessing}. In \Cref{ssec:ccRooting}, we show how to modify the algorithm of \Cref{ssec:CCproof} to obtain a rooting.

\subsection{Proof of \texorpdfstring{\Cref{thm:CCMainTheorem}}{Lg}} \label{ssec:CCproof}

There are three steps to extending \cc and proving \Cref{thm:CCMainTheorem}: (1) reducing the global memory to $O(n+m)$; (2) removing the need to know $\diam(G)$ in order to give $\ad$ as input; (3) generalizing it from trees to forests while maintaining component-stability. We address all steps separately.

\begin{enumerate}
    \item By applying \Cref{lem:ccPrePost} before executing \cc, we reduce the number of nodes in $G$ by a polynomial factor in $\ad$. This reduces the global memory to a strict $O(n+m)$.
    
    \item By employing the guessing scheme of \Cref{ssec:ccDiameter}, we perform multiple (sequential) executions of \cc. Every execution is given a doubly exponentially growing guess for $\ad$. The guessing scheme does not violate global memory $O(n+m)$ and results in a total runtime of $O(\log \diam(G))$.
    
    \item By the discussion in \Cref{ssec:ccForestsAndStability}, we can execute \cc on forests such that the runtime becomes $O(\log D)$, where $D$ is the largest diameter of any component. Moreover, when executing \cc on forests, it is component-stable.
\end{enumerate}

\subsection{\coco: Pre- and Postprocessing} \label{ssec:ccProcessing}

The aim of our \emph{preprocessing} is to reduce the number of nodes in the input graph $G$ by a factor of $\poly(\ad)$ (in fact $\ad^3$ would suffice), resulting in graph $G'$. By executing \cc on $G'$, we achieve a strict $O(n+m)$ global memory for one execution. When reducing the number of nodes, we must not disconnect the graph, and also keep the knowledge of the maximum ID inside the remaining graph.

After the connected components problem is solved on $G'$, we must extend the solution to the nodes in $G \setminus G'$ such that the solution is consistent. Extending the solution simply means informing every node in $G \setminus G'$ of \ID, which is the maximum identifier of the graph. We call this stage \emph{postprocessing}. This section is devoted to proving the following lemma.

\begin{lemma} \label{lem:ccPrePost}
	Consider a tree $G$ with $n$ nodes. The number of nodes can be reduced by a factor of $\poly(\ad)$ such that the resulting graph $G'$ remains connected, and one of the remaining nodes knows the maximum ID set $V_G \setminus V_{G'}$. If connected components is solved in $G'$, the solution can be extended to $G$. Both obtaining graph $G'$ from $G$ and extending the solution from $G'$ to $G$ takes $O(\log \ad)$ low-space \mpc rounds using $O(n+m)$ words of global memory.
\end{lemma}

Let us restate a known result that is going to be an essential tool in our preprocessing. We present a proof sketch to  explicitly reason the memory bound. 

\begin{lemma}[\cite{Czumaj2020}] \label{lem:largeIS}
	There is an $O(1)$-round sublinear local memory (component unstable) \mpc algorithm that, given a subset  $U\subseteq V$ of nodes of a graph $G=(V,E)$ with $d_G(u)=2$ for all $u\in U$, computes a subset $S\subseteq U$ that is an independent set in $G$ and satisfies $|S|\geq |U|/8$.
	The global memory used by the algorithm is $O(|U|)+O(|M|\cdot \log n)=O(n)$.
\end{lemma}

\begin{proof}[Proof Sketch]
	Consider the following random process: Each node marks itself with probability $1/2$. If a node is marked and no neighbor is marked, it joins the set $S$, otherwise it does not. The probability of a node to be marked and not having any of its neighbors marked is $1/8$. Thus, the expected size of $S$ is $|U|/8$. Further, note that this analysis still holds if the randomness for the nodes is $3$-independent. 
	$|U|$ coins that are $3$-independent can be created from a bitstring of length $O(3\cdot \log |U|)=O(\log n)$ (see \Cref{def:limitedIndependence} and \Cref{thm:seed}). 
	
	In order to deterministically compute the set $S$ we use the method of conditional expectation to compute a \emph{good} bit string. 
	For that purpose break the bitstring into $O(1)$ chunks of length at most $\delta/100\log n$. 
	Then we deterministically choose the bits on these segments such that the expected size of $S$ is remains $|U|/8$ conditioned on all already determined segments of random bits. 
	To fix one segment introduce the indicator random variable $S_v$ that equals $1$ if and only if $v\in S$. Let $\phi$ be the event that fixes to bitstring to what is already there and for $\alpha\in [n^{\delta/100}]$ let $\phi_{\alpha}$ be the event that the to be fixed segment equals $\alpha$.  Knowing the Ids of its neighbors and the already fixed part of the bitstring, each machine can for each $v\in U$ that it holds compute the values 
	$S_{v,\alpha}=\Exp[S_v \mid \psi=\alpha \wedge \phi]$. Then, nodes fix the current segment to the $\alpha_0$ such that minimizes $\sum_{v\in U}S_{v,\alpha}$. By the method of conditional expectation we have $\Exp[|S| \mid \psi_{\alpha_0}\wedge\phi]\leq |U|/8$. 
	At the end the whole bitstring is fixed and we have deterministically selected an independent set $S$ satisfying $|S|\leq |U|/8$.
	
	For an \mpc implementation, we need to be able to globally, i.e., among all machines, to agree on the good bit string obtained from the method of conditional expectations.
	For this purpose, consider an aggregation tree structure, where the machines are arranged into a (roughly) $n^{\delta/2}$-ary tree~\Cref{def:aggTreeStructure} with depth $O(1/\delta)$.
	In this tree, each machine can choose the (locally) good bit segments (conditioned on the previous segments) of $\delta/100 \log n$ bits.
	Notice that there are at most $n^{\delta/100}$ such bit segments.
	Now, we can convergecast the expected size of $S$, given a segment, to the root.
	Then, the root can decide on the \emph{good} (prefix of a) bit string.
	In each round, each machine receives a $n^{\delta/2} \cdot n^{\delta/100} \ll n^{\delta}$ bits, which fits the local memory.
	Since we have  $O(n^{1 - \delta)}$ machines, the total memory requirement to store the bits is $O(n^{\delta/2 + \delta/100 + (1 - \delta)}) = O(n)$.
	
	\begin{definition}[\cite{Vadhan12}]
		\label{def:limitedIndependence}
		For $N,M,k\in\mathbb{N}$ such that $k\leq N$, a family of functions $\mathcal{H}=\{h:[N]\to [M]\}$ is $k$-wise independent if for all distinct $x_1,\dots,x_k\in[N]$, the random variables $h(x_1),\dots,h(x_k)$ are independent and uniformly distributed in $[M]$ when $h$ is chosen uniformly at random from $\mathcal{H}$.
	\end{definition}
	
	\begin{theorem}[\cite{Vadhan12}]\label{thm:seed}
		For every $a,b,k$, there is a family of $k$-wise independent hash functions $\mathcal{H}=\{h:\{0,1\}^a\to\{0,1\}^b\}$ such that choosing a random function
		from $\mathcal{H}$ takes $k\cdot\max\{a,b\}$ random bits. \qedhere
	\end{theorem} 
\end{proof}

Next, we introduce elementary operations \rk, \con, which we use during preprocessing, and \ins, \ex, which we use during postprocessing. Operation \ins can be thought of as the reversal of \rk, and operation \ex as the reversal of \con. Recall the definition of compression and decompression from the beginning of the section.

\begin{definition} \label{def:elemOperations}
	For a degree-1 node $v$ define the following two operations.
	\begin{itemize}
		\item[$-$] $\rk(v)$: Node $v$ compresses into its unique neighbor $u$
		
		\item[$+$] $\ins(v)$: Node $v$ that underwent $\rk$ decompresses from $u$.
	\end{itemize}
	For a degree-2 node $v$ define the following two operations. 
	\begin{itemize} 
		\item[$-$] $\con(v)$: Node $v$ with neighbors $u$ and $w$ compresses into its highest ID neighbor.		
		\item[$+$] $\ex(v)$: Node $v$ that underwent $\con$ decompresses from $u$ (w.l.o.g. $\text{ID}(u)>\text{ID}(w)$).
	\end{itemize}
\end{definition}

As long as we ensure that $\ins(v)$ and $\ex(v)$ are executed on nodes which have undergone $\rk(v)$ and $\con(v)$ operations, respectively, we obtain the following observation.

\begin{observation}
	The  operations $\rk(v)$, $\ins(v$), $\con(v)$ and $\ex(v)$ can be implemented in $O(1)$ low-space \mpc rounds using $O(n+m)$ global memory on all nodes $v \in Z \subseteq V$ in parallel, if $Z$ is an independent set containing only degree-1 and degree-2 nodes.
\end{observation}

Let us introduce preprocessing and postprocessing formally. Note that preprocessing is performed directly on the input graph $G$, resulting in smaller graph $G'$. Whereas postprocessing is performed on a solved version of $G'$, i.e., with all nodes $v$ having $\id_v=\ID$, resulting in a solved version of the input graph $G$. Both of the following routines use some constant $c$ in their runtime in order to reduce the number of nodes by a factor of $\ad^c$. Initialize $G_0$ as the input graph $G$.

\begin{itemize}
	\item \pre. For $j = 0,\dots, c \log \ad$ iterations:
	\begin{enumerate}
		\item Let $H$ be the subgraph induced by all degree-2 nodes in $G_j$. Compute an independent set $Z \in H$ of size at least $|H|/8$ using \Cref{lem:largeIS}.
		\item $\con(v)$ for every $v \in Z$.
		\item $\rk(v)$ for every degree-1 node $v$. If two leaves are neighbors, perform \rk only on the higher ID one.
	\end{enumerate}
	\item \post. For $j = c \log \ad,\dots,0$ iterations:
	\begin{enumerate}
		\item ${\sf Insert}(v)$ for every node $v$ that performed $\rk(v)$ in iteration $j$ of \pre. 
		\item $\ex(v)$ for every node $v \in Z$ in iteration $j$ of \pre. 
	\end{enumerate}
\end{itemize}

\begin{proof}[Proof of \Cref{lem:ccPrePost}]
	Performing \pre takes $O(\log \ad)$ time since Steps 1--3 can be performed in constant time. Let $G$ ($G'$) be the graph before (after) performing \pre. Graph $G'$ has a $\poly(\ad)$ fraction less nodes that graph $G$ because a constant fraction of nodes in a tree have degree $\leq 2$, and we compress all degree $\leq 2$ nodes in the graph in each iteration. 
	
	Performing \post simply reverses \pre while extending the current solution, so it has the same runtime as \pre. Both \pre and \post can be implemented in low-space \mpc using $O(n+m)$ global memory, since other that \Cref{lem:largeIS}, nodes only exchange constant sized messages with their neighbors, and don't store anything non-constant in local memory.
\end{proof}

\subsection{\coco: Removing Knowledge of the Diameter} \label{ssec:ccDiameter}

Our \cc algorithm requires value $\ad \in [D,n^{\delta/8}]$ as input, where $D$ denotes the diameter of the input tree. We show that we don't actually need to know $D$ in order to give value $\ad$ as input. We achieve this by sequentially executing \cc with a doubly exponential guess for $\ad$ as follows.

\begin{align*}
	\ad = \ad_1,\ad_2,\dots,\ad_{\log \log n^{\delta/8}} = 2^{2^1},2^{2^2},\dots,2^{2^{\log \log n^{\delta/8}}}~.
\end{align*}

We proceed with the next guess only if the previous has failed to terminate after a runtime of $O(\log \ad_i)$. Detecting a failure within the given runtime, and making sure that a wrong guess does not violate memory constraints is a delicate affair, and is discussed in a separate paragraph. If all of our guesses fail, it must be that $D > n^{\delta/8}$. In this case, we can run the deterministic $O(\log n)$ time connected components algorithm of Coy and Czumaj \cite{ccderandom} with the requirement that all nodes output the maximum ID of the component (their algorithm is component-stable for the same reasons the algorithm in this paper is). Hence, we can safely assume that $D \leq n^{\delta/8}$. Assuming that failure detection can be performed within $O(\log \ad_i)$ and the given memory constraints, we show that the algorithm terminates successfully for some guess $\ad$, and that the runtime resulting from our guessing scheme is acceptable. Eventually for some guess $l$, it holds that $D \leq \ad_l$ and $\ad_{l'}<D$ for all $l'<l$. In particular, it holds that $D \leq \ad_l \leq D^2$. Since $\ad_l \in [D,\min(n^{\delta/8},D^2)] \subseteq [D,n^{\delta/8}]$, \cc will terminate successfully for $\ad_l$. Our guessing scheme results in a runtime of at most
\begin{align*}
	\sum_{i=1}^{l} O(\log \ad_i) &\leq \left( \sum_{i=0}^{\infty} \frac{1}{2^i} \right) O(\log D) + O(\log D^2) = O(\log D)~.
\end{align*} 

\paragraph{Detecting Failure.} Let us consider the case when our guess $\ad$ for the diameter is $<D$. It is possible for the algorithm to terminate even with a wrong diameter guess. However, we want to show that when the guess is wrong, we are able to detect it in $O(\log \ad)$ time even when the algorithm has not terminated. We also want to show that using a wrong diameter guess does not violate our memory constraints. 

Our aim is to show that there exists a constant $c$ such that if, after $c \log \ad$ rounds, the algorithm is unsuccessful (failed), we can manually terminate the execution and move on to the next diameter guess. The algorithm can be seen as unsuccessful, if after a constant number of phases, the graph is larger than a singleton (\Cref{lem:lConstant}). Note that failure cannot be detected during preprocessing, since there we simply free an appropriate amount of memory for the actual algorithm. The exact number of phases can be deduced from \Cref{lem:lConstant}. Phases consist of algorithms \clst and \cp which both have a runtime of $O(\log \ad)$ by \Cref{lem:CompressLightSubTrees} and \Cref{lem:CompressPaths}, respectively. The exact constant in \Cref{lem:CompressLightSubTrees} can be computed from \Cref{lem:correctnessLemmaLight,lem:pathShortening,lem:mainProbing}. If the diameter guess is wrong, it will simply result in removing too small subtrees, and possibly a graph larger than a singleton being left after the phases. The exact constant in \Cref{lem:CompressPaths} can be computed from its proof in \Cref{ssec:ccSinglePhasePaths}. If the diameter guess is wrong, it may still result in \cp terminating successfully. However, it may also result in nodes not learning the whole path they are in, which we can detect and manually terminate the execution. 

It is left to show that for a wrong diameter guess, the local memory $O(n^\delta)$ and global memory $O(n+m)$ are not violated. Both cases are surprisingly straightforward. The former holds since all of the local memory arguments of the section are independent of $\ad$ (they only use its upper bound $n^{\delta/8}$). The latter holds by performing \pre for $3 \log \ad$ iterations before every execution, since all of our lemmas use at most $O((n+m)\cdot \ad^3)$ global memory. 

\subsection{\coco: Forests and Stability} \label{ssec:ccForestsAndStability}

When executing algorithm \cc and the scheme developed in \Cref{ssec:CCproof} on forests instead of trees, we have to consider how the disjoint components can affect each other with regards to runtime, memory, and component-stability (\Cref{def:componentStability}).

\begin{definition}[Component-stability, \cite{componentstable}] \label{def:componentStability}
	A randomized \mpc algorithm $A_{\mpc}$ is component-stable if its output at any node $v$ is entirely, deterministically, dependent on the topology and IDs (but independent of names) of $v$’s connected component (which we will denote $CC(v)$), $v$ itself, the exact number of nodes $n$ and maximum degree $\Delta$ in the \emph{entire} input graph, and the input random seed $\mathcal{S}$. That is, the output of $A_{\mpc}$ at $v$ can be expressed as a deterministic function $A_{\mpc}(CC(v),v,n,\Delta,\mathcal{S})$. A deterministic \mpc algorithm $A_{\mpc}$ is component-stable under the same definition, but omitting dependency on the random seed $\mathcal{S}$.
\end{definition}

If we were to execute \cc on a forest, nodes from disjoint components would never communicate with each other, rendering the runtime and memory arguments local. Hence, the algorithm is compatible with forests. In the scheme developed in  \Cref{ssec:CCproof} nodes from disjoint components communicate with each other only during preprocessing, when we employ the $O(1)$ time independent set algorithm of \Cref{lem:largeIS}. Since the independent set is used to reduce the number of nodes \emph{globally}, all of the runtime and memory arguments are still compatible with forests, as long as the given value $\ad$ is in $[D,n^{\delta/8}]$, where $D$ is the largest diameter of any component. 

What is left to argue is that if the input graph is a forest, \cc and the scheme developed in \Cref{ssec:CCproof} are component-stable. This however follows directly from the stability definition (\Cref{def:componentStability}) and our problem definition (\Cref{def:CCproblem}), because we require nodes to output the maximum identifier of their component, which is fully independent of other components. 

\subsection{Computing a Rooted Forest} \label{ssec:ccRooting}
In this section, we show how to use the connected components  algorithm of \Cref{thm:CCMainTheorem}, with minor adjustments, to root a forest. We prove the following. 

\thmTreeRooting*

First, we execute the algorithm of \Cref{thm:CCMainTheorem} in order for every node to learn the maximum ID of its component. Using this knowledge, we execute a modified version of the same algorithm that roots (in a component-stable way) every component towards the single node (per component) that is left after $\ell$ phases. The modifications are the following.
\begin{enumerate}
	\item Redefine compression and decompression as follows 
	\begin{itemize}
		\item Compress $X$ into $v$: remove $X$ (and its incident edges) from the graph. For any edge $\{x,y\}$ with  $x\in X$ and $v\neq y \notin X$ we introduce a new edge $\{v,y\}$. 
		\item Decompress $X$ from $v$: Decompress $X$ from $v$: add $X$ (and its incident edges) back to the graph. If set $X$ is a subtree, orient the revived edges towards the single node to which the subtree is attached to. If set $X$ is a path, orient the revived edges in the same direction as edge $\{v,y\}$ that was added during the compression step of $X$ into $v$.
	\end{itemize}
	\item In the derandomization of \Cref{lem:largeIS}, the process is run independently on each connected component. For each component, we use an aggregation tree (recall \Cref{def:aggTreeStructure}) that consist of nodes only in the corresponding component with the maximum ID node as a root. Then, the good bitstrings can be determined through the independent aggregation trees.
\end{enumerate}

\begin{proof}[Proof of \Cref{thm:RootingMainTheorem}]
	The first modification to the algorithm ensures a rooted forest. We prove it by induction, with the base case being a rooted graph $G_{\ell} = \{v\}$.  If graph $G_i$ is rooted in the beginning of a reversal phase $i$, \dcp extends the rooting of (some) single edges to paths, and \dclst extends the rooting of (all) leaf nodes to subtrees, resulting in a rooted graph $G_{i-1}$ (recall that the reversal phase indices are in decreasing order). The same exact logic also holds for \post, where we extend the rooting to the nodes that were removed during \pre. 
	
	The second modification ensures component-stability. Since the maximum ID node of each component is responsible only for its own component, we can choose the bitstring independently of the other components and create the broadcast tree (like in \Cref{lem:largeIS}) for each component separately and independently. 
	Then, the orientation of the rooting (i.e., which node will become the root) only depends on the topology and the maximum ID of the component, making it independent of the other components.
\end{proof}

\section{Solving \lcl Problems}
\label{sec:LCLSection}
In this section, we show a useful application of our forest rooting algorithm. In particular, we show that all problems contained in a wide class of problems that has been heavily studied in the distributed setting, called Locally Checkable Labelings (\lcl{}s), can be solved in $O(\log D)$ deterministic rounds.

Informally, \lcl{}s are a restriction of a class of problems called \emph{locally checkable problems}. These problems satisfy that, given a solution, it is possible to check whether the solution is correct by checking the constant radius neighborhood around each node separately. Examples of these problems are classical problems such as maximal independent set, maximal matching, and $(\Delta+1)$-vertex coloring, but also more artificial problems, such that the problem of orienting the edges of a graph such that every node must have an odd number of outgoing edges. 

The restriction that is imposed on locally checkable problems to obtain the class of \lcl{}s is to require that the number of possible input and output labels that are required to define the problem must be constant, and moreover only graphs of bounded degree are considered. In the distributed setting, and in particular in the \local model of distributed computing, \lcl{}s have been extensively studied, see, e.g., \cite{balliu2020, BBOS18, CP19, BHKLOS18, ChangKP19, BCMOS21, Chang2020, Grunau2022}. In particular, the imposed restriction makes it possible to prove very interesting properties on them, and to develop generic techniques to solve them.  For example, we know that, if we restrict to forests, there are \lcl{}s that can be solved in $O(1)$ rounds, there are \lcl{}s that require $\Omega(\log^* n)$ rounds, but we also know that there is nothing in between, even if randomness is allowed (e.g., there are no \lcl{}s with complexity $\Theta(\sqrt{\log^* n})$). Interestingly, techniques that have been developed to study \lcl{}s have then often been extended and used to understand locally checkable problems in general (that is, problems that are not necessarily \lcl{}s).

Since the \local model is very powerful, and allows to send arbitrarily large messages, any solvable problem can be solved in $O(D)$ rounds. 
In this section, we provide an \mpc{} algorithm for solving any solvable \lcl problem on forests. The algorithm is deterministic, component-stable, and runs in $O(\log D)$ time in the low-space \mpc model using $O(n + m)$ words of global memory (formal statement in \Cref{thm:LCLSolver}). Moreover, our algorithm can be used \emph{even for unsolvable \lcl{}s}, that is, problems for which there exists some instance in which they are unsolvable. Hence, given any \lcl, our algorithm produces a correct output on any instance that admits a solution, and it outputs "not solvable" on those instances where the \lcl is not solvable. 

\medskip

\noindent\textbf{\Cref{thm:LCLSolver}.}  
	\emph{All \lcl problems on forests with maximum component diameter $D$ can be solved in $O(\log D)$ time in the low-space \mpc model using $O(n+m)$ words of global memory. The algorithm is deterministic and does not require prior knowledge of $D$.}

\medskip

Notice that it is enough to prove \Cref{thm:LCLSolver} for rooted forests, since we can first root the forest by spending the same runtime and memory (\Cref{thm:CCMainTheorem}), and then solve the \lcl. 

The remaining of the section is structured as follows: we start by giving a formal definition of \lcl{}s (\Cref{subsec:lcls}); we proceed by providing a high-level overview of our algorithm (\Cref{ssec:lcloverview}); then, we provide the definition of the concept of ``compatibility tree'', that will be useful later (\Cref{ssec:compTree}); in Section \ref{sec:lclsolver} we give the explicit algorithm, called \sr; in \cref{sec:cs,sec:gs,sec:cs,sec:cp,sec:dcp,sec:dcs} we show some properties of the subroutines used in \sr, and we bound the time complexity of each of them; finally, we put things together and prove the main theorem of this section in \Cref{sec:proofsec}. 

We note that, for the sake of simplicity, algorithm \sr is described for trees, but we will show in \Cref{sec:proofsec} that it can also be executed on forests.

\subsection{Locally Checkable Labelings}\label{subsec:lcls}

Locally Checkable Labeling (\lcl) problems have been introduced in a seminal work of Naor and Stockmeyer \cite{naor-stockmeyer}.  The definition they provide restricts attention to problems where the goal is to label nodes (such as vertex coloring problems), but they remark that a similar definition can be given for problems where the goal is to label edges (such as edge coloring problems). A modern way to define \lcl problems that captures both of the above types of problems (and combinations thereof) consists of labeling of \emph{half-edges}, i.e., pairs $(v,e)$ where $e$ is an edge incident to vertex $v$. Let us first formally define half-edge labelings, and then provide this modern \lcl problem definition.

\begin{definition}[Half-edge labeling]\label{def:halfedge}
	A \emph{half-edge} in a graph $G = (V,E)$ is a pair $(v,e)$, where $v \in V$, and $e=\{u,v\} \in E$.
	We say that a half-edge $(v,e)$ is incident to some vertex $w$ if $v = w$.
	We denote the set of half-edges of $G$ by $H = H(G)$.
	A \emph{half-edge labeling} of $G$ with labels from a set $\Sigma$ is a function $g \colon H(G) \to \Sigma$.
\end{definition}

We distinguish between two kinds of half-edge labelings: \emph{input labelings}, that are labels that are part of the input, and \emph{output labelings}, that are provided by an algorithm executed on input-labeled instances. Throughout the paper, we will assume that any considered input graph $G$ comes with an input labeling $\ginn \colon H(G) \to \sinn$ and will refer to $\sinn$ as the \emph{set of input labels}; if the considered \lcl problem does not have input labels, we can simply assume that $\sinn = \{\bot\}$ and that each half-edge is labeled with $\bot$.

Informally, \lcls are defined on bounded-degree graphs, where each node may have in input a label from a constant-size set $\sinn$ of labels, and must produce in output a label from a constant-size set $\sout$ of labels. Then, an \lcl is defined through a set of locally checkable constraints that must be satisfied by all nodes.
\begin{definition}[\lcl] \label{def:lcl}
	An \lcl problem $\Pi = (\sinn,\sout,C,r)$ is defined as follows:
	\begin{itemize}
		\item $\sinn$ and $\sout$ are sets of constant size that represent, respectively, the possible input and output labels.
		\item The parameter $r$ is a constant called \emph{checkability radius of $\Pi$}.
		\item $C$ is a set of constant size, containing allowed neighborhoods. Each element $c_i = (G_i, v_i)$ of $C$, where $G_i = (V_i,E_i)$, is such that:
		\begin{itemize}
			\item $G_i$ is a graph satisfying that $v_i \in V_i$ and that the eccentricity of $v_i$ in $G_i$ is at most $r$;
			\item Every half-edge of $G_i$ is labeled with a label in $\sinn$ and a label in $\sout$.
		\end{itemize}
	\end{itemize}
\end{definition}

\begin{definition}[Solving an \lcl] \label{def:solveLCL}
	In order to solve an \lcl on a given graph $G = (V,E)$ where to each element $(v,e) \in V \times E$ is assigned an input label from $\sinn$, we must assign to each element $(v,e) \in V \times E$ an output label from $\sout$ such that, for every $v \in V$, it holds that $(G_r(v),v) \in C$, where $G_r(v)$ is the subgraph of $G$ induced by nodes at distance at most $r$ from $v$ and edges incident to at least one node at distance at most $r-1$ from $v$.
\end{definition}

\begin{example}[Maximal Independent Set]
    In the maximal independent set problem, the goal is to select an independent set of nodes that cannot be extended. That is, selected nodes must not be neighbors, and non-selected nodes must have at least one neighbor in the set.
    
    For this problem, $\sinn = \{\bot\}$. Then, we can use two possible output labels, $1$ to indicate nodes that are in the set, and $0$ to indicate nodes that are not in the set. Hence,  $\sout = \{0,1\}$. Finally, we need to define $r$ and $C$. For this problem, it is sufficient to pick $r=1$. In $C$, we put all possible pairs $(G,v)$ satisfying the following:
    \begin{itemize}
        \item $G$ is a star centered at $v$;
        \item $G$ has at most $\Delta$ leaves (and there can be $0$ leaves);
        \item For each node in $G$, either all incident half-edges are output labeled $0$, or all incident half-edges are output labeled $1$;
        \item If $v$ is labeled $1$, then all leaves are labeled $0$;
        \item If $v$ is labeled $0$, then at least one leaf is labeled $1$.
    \end{itemize}
    
    Hence, the idea is that MIS can be checked by just inspecting the radius-$1$ neighborhood of each node, which is a star, and we just list all stars that are valid.
\end{example}

While \cref{def:lcl} gives an easy way to define problems, such a definition is not the most convenient for proving statements about \lcl{}s. In order to make our proofs more accessible, we consider an alternative definition of \lcls, called node-edge formalism. It is known that, on trees and forests, any \lcl defined as in \cref{def:lcl} can be converted, in a mechanical way, into an \lcl described by using this formalism, such that the obtained \lcl has the same asymptotic complexity of the original one \cite{BCMOS21}.

\begin{definition}[Node-edge-checkable \lcl]\label{def:nodeedge}
	In this formalism, a problem $\Pi$ is a tuple $(\sinn,\sout,C_V,C_E)$ satisfying the following:
	\begin{itemize}
		\item As before, $\sinn$ and $\sout$ are sets of constant size that represent, respectively, the possible input and output labels;
		\item $C_V$ and $C_E$ are both sets of multisets of pairs of labels, where each pair is in $\sinn \times \sout$, and multisets in $C_E$ have size $2$.
	\end{itemize}
\end{definition}

\begin{definition}[Solving a node-edge checkable \lcl] \label{def:solveNodeEdgeLCL}
	Solving an \lcl given in this formalism means that we are given a graph $G = (V,E)$ where to each element $(v,e) \in V\times E$ is assigned a label $i_{v,e}$ from $\sinn$, and to each element $(v,e) \in V\times E$ we must assign a label $o_{v,e}$ from $\sout$ such that:
	\begin{itemize}
		\item For every node $v \in V$ it holds that the multiset $M_v = \{(i_{v,e},o_{v,e}) ~|~ e \text{ is incident to } v\}$ satisfies $M_v \in C_V$;
		\item For every edge $e \in E$ it holds that the multiset $M_e = \{(i_{v,e},o_{v,e}) ~|~ e \text{ is incident to } v\}$ satisfies $M_e \in C_E$.
	\end{itemize}
\end{definition}
Hence, in the node-edge checkable formalism, we are given a graph where each half-edge (that is, an element from $V \times E$) is labeled with a label from $\sinn$, the task is to label each half-edge from a label from $\sout$, and the \lcl constraints are expressed by listing tuples of size at most $\Delta$ representing allowed configurations for the nodes, and tuples of size $2$ representing allowed configurations for the edges. In \cite{BCMOS21} it has been shown that any \lcl $\Pi$ defined on trees or forests can be converted into a node-edge checkable \lcl $\Pi'$ satisfying that the complexity of $\Pi$ and $\Pi'$ differ only by an additive constant. Hence, for the purposes of this work, we can safely restrict our attention to node-edge-checkable \lcl{}s.

\begin{example}[Maximal Independent Set]
    Sometimes, defining an \lcl{} in the node-edge checkable formalism is non-trivial. MIS is an example of problems in which the conversion requires a bit of work (it can be done mechanically, though, as shown in \cite{BCMOS21}). We hence use MIS as an example for this formalism.
    
    As before, $\sinn = \{\bot\}$, and hence when listing the elements in $C_V$ and $C_E$ we will not specify the input labels. This time, it is not actually possible to use just $2$ labels as output. In fact, we define $\sout = \{0, 1, \mathrm{P}\}$. Then, $C_V$ contains all multisets of size at most $\Delta$ satisfying that:
    \begin{itemize}
        \item All elements are $1$, or
        \item one element is $\mathrm{P}$ and all the others are $0$.
    \end{itemize}
    Then, $C_E = \{\{1,0\},\{1,\mathrm{P}\},\{0,0\}\}$. In other words, nodes in the MIS output $1$ on all their incident half-edges, nodes not in the MIS output $\mathrm{P}$ on one incident half-edge and $0$ on all the others. The label $\mathrm{P}$ is used to prove maximality. That is, nodes not in the set must point to one neighbor in the set by using the label $\mathrm{P}$. In fact, on the edge constraint, $\mathrm{P}$ is only compatible with $1$. Observe that, given a solution for the standard MIS problem, a solution for this variant can be produced with just one round of communication.
\end{example}

\subsection{Overview}\label{ssec:lcloverview}
On a high level, our algorithm works as follows. We describe it from the point of view of a single node, and for a single tree of the forest. Firstly, we root the tree, obtaining that each node knows the edge connecting it to its parent. Then, the algorithm proceeds in phases, and in total the number of phases is going to be a constant that depends on the amount of memory available to the machines. In each phase, we \emph{compress} the tree into a smaller tree, as follows:
\begin{itemize}
    \item all subtrees containing less than a fixed amount of nodes are compressed to their root;
    \item all paths are compressed into a single edge.
\end{itemize}
Each phase is going to require $O(\log D)$ time. In other words, this part of our algorithm works similar to the standard rake-and-compress algorithm.
Moreover, while compressing the tree, we maintain some information about the \lcl{} that we are trying to solve. This information is called \emph{compatibility tree}.

The compatibility tree, for each node and for each edge, keeps track of the possible configurations that they can use. At the beginning, for each node, these configurations correspond to the configurations in $C_V$ that are compatible with the given input, and for each edge, these configurations correspond to the configurations in $C_E$ that are compatible with the given input. 

When compressing a subtree into a single node, we update the list of the configurations usable on that node, in such a way that each configuration satisfies the following: if the node uses it, then it is possible to assign a labeling on the subtree compressed into that node, in such a way that, for each node and edge in the compressed subtree we use only configurations allowed by the compatibility tree before the compression.

Similarly, when compressing a path, for the new edge that we add, we store a list of configurations satisfying that, if we label the first and last half-edge of the removed path with the labels of the configuration, then we can complete the compressed path by only using configurations allowed by the compatibility tree before the compression.

At the end, we obtain that the whole tree is recursively compressed on a single node $v$. If the compatibility tree does not allow any configuration for $v$, then we know that the \lcl{} is unsolvable. Otherwise, we can pick an arbitrary configuration allowed by the compatibility tree and assign it to $v$. By performing this operation, we know that we can safely put back the paths and subtrees that were compressed on $v$ and have the guarantee that we can label them using only allowed configurations. Hence, we again proceed in phases, where we put back compressed paths and subtrees in the opposite order in which they have been compressed, and each time we assign labels allowed by the compatibility tree. At the end, we obtain that the whole tree is labeled correctly, and hence the \lcl is solved.

\subsection{Compatibility Tree} \label{ssec:compTree}
A compatibility tree is an assignment of sets of allowed configurations to nodes and edges, where this time configurations are not just multisets, but they are tuples. In other words, we may allow a node to use a configuration, but only if the labels of that configuration are used in a very specific order.
\begin{definition}[Compatibility Tree] \label{def:compTree}
    A compatibility tree of a tree $G = (V,E)$ is a pair of functions $\phi$ and $\psi$, where $\phi$ maps each node $v \in V$ into a set of tuples of size at most $\Delta$, and $\psi$ maps each edge $e \in E$ into a set of tuples of size $2$.
\end{definition}

In order to specify how the compatibility tree is initialized, it is useful to first assign an order to the edges incident to each node, and to the nodes incident to each edge. This ordering is called port numbering assignment. Observe that an arbitrary port numbering assignment can be trivially computed in $1$ round of communication.
\begin{definition}[Port Numbering]\label{def:portnumbering}
    A node port numbering is a labeling of every half-edge satisfying that, for each node $v$, half-edges incident to $v$ have pairwise distinct values in $\{1,\ldots,\deg(v)\}$. An edge port numbering is a labeling of every half-edge satisfying that, for each edge $e$, half-edges incident to $e$ have pairwise distinct values in $\{1,2\}$. A port numbering is the union of a node port numbering and an edge port numbering.
\end{definition}

Assume that the tree $G$ is already provided with a port numbering. The compatibility tree of $G$ is initialized as follows. 
For each node $v$, $\phi(v) = \{  (\ell_1,\ldots,\ell_{\deg(v)}) ~|~ \{(i_1,\ell_1),\ldots,(i_{\deg(v)},\ell_{\deg(v)})\} \in C_V \}$, where $i_j$ is the input assigned to the half-edge incident to $v$ with node port number $j$.
For each edge $e$, $\psi(e) = \{  (\ell_1,\ell_2) ~|~ \{(i_1,\ell_1),(i_2,\ell_2)\} \in C_E \}$, where $i_j$ is the input assigned to the half-edge incident to $e$ with edge port number $j$. In other words, we initialize $\phi$ and $\psi$ with everything that is allowed by the constraints of the problem, in all possible orders that are compatible with the given input.

We can observe that, by construction, $\phi$ and $\psi$ still encode the original problem. In other words, we can now forget about $C_V$ and $C_E$, and try to find a labeling assignment that is valid according to $\phi$ and $\psi$. We make this observation more formal in the following statement.
\begin{observation}\label{obs:comp}
    The \lcl problem $\Pi$ is solvable if and only if there is a labeling $g_{\mathrm{out}} : H \rightarrow \sout$ that solves $\Pi$ that satisfies that:
    \begin{itemize}
        \item For each node $v$, let $\ell_j$ be the label assigned by $g_{\mathrm{out}}$ to the half-edge incident to $v$ with port number $j$. It must hold that $(\ell_1,\ldots,\ell_{\deg(v)}) \in \phi(v)$.
        \item For each edge $e$, let $\ell_j$ be the label assigned by $g_{\mathrm{out}}$ to the half-edge incident to $e$ with port number $j$. It must hold that $(\ell_1,\ell_2) \in \psi(e)$.
    \end{itemize}
\end{observation}

On a high level, when compressing a subtree into a node $v$, we will redefine $\phi(v)$ and discard some tuples. The discarded tuples are the ones satisfying that, if node $v$ uses such a configuration, there is no way to complete the labeling of the subtree in a valid way.
Similarly, when compressing a path, we will define $\psi(e)$, where $e$ is the new (virtual) edge that we use to replace the path, in such a way that, if $\psi(e)$ contains the tuple $(\ell_1,\ell_2)$ and we label the first half-edge of the compressed path with $\ell_1$ and the last half-edge with $\ell_2$, then we can correctly complete the labeling inside the path. Here it should become clear why we use tuples and not just multisets: it may be that $v$ can use a label on the half-edge connecting it to one child (because that subtree can be competed by starting with that label), but the same label cannot be used on the half-edge connecting $v$ to a different child. A similar situation could happen on a compressed path: it could be that it is possible to label $\ell_1$ the half-edge connecting the first endpoint to the path and $\ell_2$ the half-edge of the second endpoint to the path, but not vice versa.

In the algorithm, we will solve the problem $\Pi$ in the tree obtained by compressing some subtrees into single nodes, and some paths into single edges.
We now formally define what it means to partially solve an \lcl{} $\Pi$ w.r.t.\ a compatibility tree $(\phi, \psi)$. Observe that, in a tree $G$ obtained after performing some compression steps, a node $v$ may have a degree that is smaller than the size of the tuples given by $\phi(v)$, that always have size equal to the original degree of $v$, denoted by $\mathrm{origdeg}(v)$, and hence the ports incident to $v$ may be just a subset of $\{1,\ldots,\mathrm{origdeg}(v)\}$. 
\begin{definition}[Partially solving an \lcl w.r.t.\ the compatibility tree]\label{def:partialsol}
    Let $G$ be a tree, and let $(\phi, \psi)$ be a compatibility tree for $G$. A solution for $\Pi$ that is correct according to $\phi$ and $\psi$ is a labeling $g_{\mathrm{out}}$ satisfying that:
    \begin{itemize}
         \item For each node $v$, let $\mathrm{origdeg}(v)$ be the size of the tuples given by $\phi(v)$, and let $P(v) \subseteq \{1,\ldots,\mathrm{origdeg}(v)\}$ be the subset of ports of $v$ that are present in $G$. For each $j \in P(v)$,  let $\ell_j$ be the label assigned by $g_{\mathrm{out}}$ to the half-edge incident to $v$ with port number $j$. There must exist labels $\ell_k$, for all $k \in  \{1,\ldots,\mathrm{origdeg}(v)\} \setminus P(v)$, such that it holds that $(\ell_1,\ldots,\ell_{\mathrm{origdeg}(v)}) \in \phi(v)$.
        \item For each edge $e$, let $\ell_j$ be the label assigned by $g_{\mathrm{out}}$ to the half-edge incident to $e$ with port number $j$. It must hold that $(\ell_1,\ell_2) \in \psi(e)$.
    \end{itemize}
\end{definition}
In other words, solving the \lcl in the tree obtained by performing some compression steps, means to pick, for each node, a configuration allowed by $\phi$, in such a way that all edges that are still present have a configuration allowed by $\psi$.

\subsection{The Algorithm} \label{sec:lclsolver}
Let $\Pi = (\sinn, \sout, C_V, C_E)$ be the considered \lcl problem, and let $G_0 = G$ denote a rooted input tree with root $r$. The high-level idea of our approach is to first initialize $\phi_0$ and $\psi_0$ as the functions $\phi$ and $\psi$ shown in \Cref{ssec:compTree}. Then, we perform the following distinct parts.

\paragraph{Steps 1--2 of \sr.} 
From $G_0$, we iteratively derive a sequence $G_1, G_2, \dots, G_t$ of smaller trees until eventually, for some $t=O(1)$, it holds that $G_t$ consists of a single node (the root $r$). In the meanwhile, we also update the compatibility tree, and compute $\phi_j$ and $\psi_j$ for all $0 < j \le t$. The sequence is derived such that $G_j$ ($0 < j \leq t$) is obtained from $G_{j-1}$ by first compressing all subtrees of size $\leq n^{\delta/2}$ into their respective roots (we refer to the roots of the subtrees and not (necessarily) the actual root node $r$), and then compressing all paths into single edges. Throughout the sequence of compatibility trees, we maintain the following property, which we prove in Lemmas \ref{lem:cs-compat} and \ref{lem:cp-compat}.

\begin{claim} \label{claim:gi}
	Let $1 \leq j \leq t$.
	If there exists a correct solution for $G_{j-1}$ according to $\phi_{j-1}$ and $\psi_{j-1}$ (w.r.t.\ \Cref{def:partialsol}), then there exists also a correct solution for $G_{j}$  according to $\phi_{j}$ and $\psi_{j}$ (w.r.t.\ \Cref{def:partialsol}). Moreover, given any correct solution for $G_j$, we can transform it into a correct solution for $G_{j-1}$.
\end{claim}

\paragraph{Steps 3--4 of \sr.}
For all $0 \le i \le t$, we define $\dot{G}_i$ to be $G_i$ where to each node $v$ is assigned a configuration $c(v) \in \phi_i(v)$ in such a way that the assignment $c$ induces a labeling $g_{\mathrm{out}}$ that is correct according to $\phi_{i}$ and $\psi_{i}$ (w.r.t.\ \Cref{def:partialsol}).
The sequence is derived such that $\dot{G}_i$ is obtained from $\dot{G}_{i+1}$ by decompressing the subtrees and paths that were compressed when $G_{i+1}$ was obtained from $G_{i}$, and simultaneously solving the problem, which is possible by Claim \ref{claim:gi}. Finally, the solution on $\dot{G}_0$ is a solution for $\Pi$ on $G$ by \Cref{obs:comp}.

\paragraph{The Algorithm.}

Let us now formally present the algorithm \sr (along with its subroutines) that solves any \lcl on rooted trees in $O(\log D)$ time. Note that the subroutines are state-changing functions, i.e., they modify their input graphs.

\falgo{$\sr(\Pi, G(V,E))$}{
	\item Initialize $\phi_0$ and $\psi_0$ according to \Cref{ssec:compTree}.
	\item Initialize phase counter $j \larr 0$. Repeat the following until the graph is a singleton.
	\begin{enumerate}
		\item $\css(G)$
		\item $\gs(G)$
		\item $\cs(G)$ \hfill \\// compress all subtrees of size $\leq n^{\delta/2}$ into single nodes, also compute $\phi_{j+1}$ 
		\item $\acp(G)$ \hfill \\// compress all paths into single edges, also compute $\psi_{j+1}$
		\item Update $j \larr j + 1$
	\end{enumerate}
	\item Set $c(v)$ to be an arbitrary element of $\phi_j(v)$, where $v$ is the obtained singleton.\\// the graph $\dot{G}_j = \dot{G}$ is obtained
	\item Initialize repetition counter $k \larr j$. Repeat the following while $k\geq0$.
	\begin{enumerate}
		\item $\dcp(\dot{G})$ \hfill \\// decompress the paths from phase $k$
		\item $\dcs(\dot{G})$ \hfill \\// decompress the subtrees from phase $k$, the graph $\dot{G}_{k-1}$ is obtained
		\item Update $k \larr k - 1$
	\end{enumerate}
}

Observe that the phase counter is incremented, while the repetition counter is decremented, which is inline with the indexing used in the previous high-level overview. Next, we give a brief introduction to the subroutines, before defining them formally in the following subsections. Recall that $T(v)$ denotes the subtree that is rooted at a node $v$ such that $v$ belongs to $T(v)$ and $G \setminus T(v)$ is connected.

\begin{itemize}
	\item {\ul{$\css$}}: Every node $v$ learns either the exact size of $T(v)$ or that $|T(v)| > n^{\delta/2}$. In the former case, $v$ marks itself as \emph{light}, and in the latter case, as \textit{heavy}. If a heavy node has a light child, it remarks itself as a \textit{local root}.
	\item {\ul{$\gs$}}: Every local root $v$ learns $T(u)$ for every light child $u$.
	\item {\ul{$\cs$}}: Every local root $v$ checks, for each half-edge $h$ connecting it to a light child $u$, what are the possible labelings of $h$ that allow to complete the labeling of $T(u)$ in such a way that it is valid according to $\phi$ and $\psi$. Then,  $v$ compresses all of these trees into itself, and updates $\phi$ in such a way that any \lcl solution on the remaining graph can be extended to a solution on $T(u)$ for every light child $u$. This is done according to the computed possible labelings of the half-edges.
	\item {\ul{$\acp$}}: For every path $P$ with some endpoints $u$ and $w$ (both have either degree 1 or $\geq 3$), compress $P$ into a new edge $\{u,w\}$, and assign an arbitrary edge port numbering to this edge. The value of $\psi(\{u,w\})$ is then defined in a way that any \lcl solution on the edge $\{u,w\}$ can be extended to a solution on $P$. After performing \cp, there are no degree-2 nodes left, which will be crucial for the analysis.

	\item {\ul{$\dcp$}}: The \lcl problem on the input graph is solved. In repetition $k$, we decompress all paths that were compressed during phase $k$. While decompressing, we extend the solution to the paths.
	\item {\ul{$\dcs$}}: The \lcl problem on the input graph is solved. In repetition $k$, every local root of phase $k$ decompresses all subtrees that it compressed during phase $k$. While decompressing, we extend the solution to the subtrees.
\end{itemize}

The following sections are rather self-contained and correspond to a specific subroutine that is called by \sr (in the order they are called). They are written from a node's point of view, with the proofs intertwining correctness, runtime, and \mpc details. 

\subsection{Solving \lcls: \css} \label{sec:css}

During the execution of the subroutine, every node $v$ maintains the following variables:
\begin{itemize}
	\item $i$: iteration counter
	\item $s(v)$: size of $T(v)$ until depth $2^i$ ($v$ is at depth 0)
	\item $C(v)$: set of all descendant nodes (if any) of $v$ at depth $2^i$ ($v$ is at depth 0).
\end{itemize}

The aim of the subroutine is to detect all heavy nodes, i.e., nodes which have a subtree of size $> n^{\delta/2}$ rooted at them. This can be thought of as a preprocessing step for \gs. All nodes are initially marked as active.

\falgo{$\css(G)$}{
	\item[] \hspace{-6mm} Each node $v$ initializes: $C(v) \larr$ set of children of $v$ in $G$, $s(v) \larr |C(v)|$, and $i \larr 0$.
	\item Repeat the following steps until $C(v) = \emptyset$ for every node $v$.
	\begin{enumerate}[leftmargin=*]
		\item If $v$ is active and all $u \in C(v)$ are also active,
		$v$ updates 
		\begin{itemize}
			\item[--] $C(v) \larr \bigcup_{u \in C(v)} C(u)$
			\item[--] $s(v) \larr s(v) + \sum_{u \in C(v)} s(u)$.
		\end{itemize}
		\vspace{1mm}
		Otherwise, $v$ marks itself as heavy, and it becomes inactive.
		\item If $s(v) > n^{\delta/2}$, $v$ marks itself as heavy, and it becomes inactive. Heavy nodes update $C(v) \larr \emptyset$.
		\item All active nodes update $i \larr i + 1$.
	\end{enumerate}
	\item[] \hspace{-6mm} Non-heavy nodes marks themselves as light.
}

\begin{figure}
	\centering
	\includegraphics[width=0.7\textwidth]{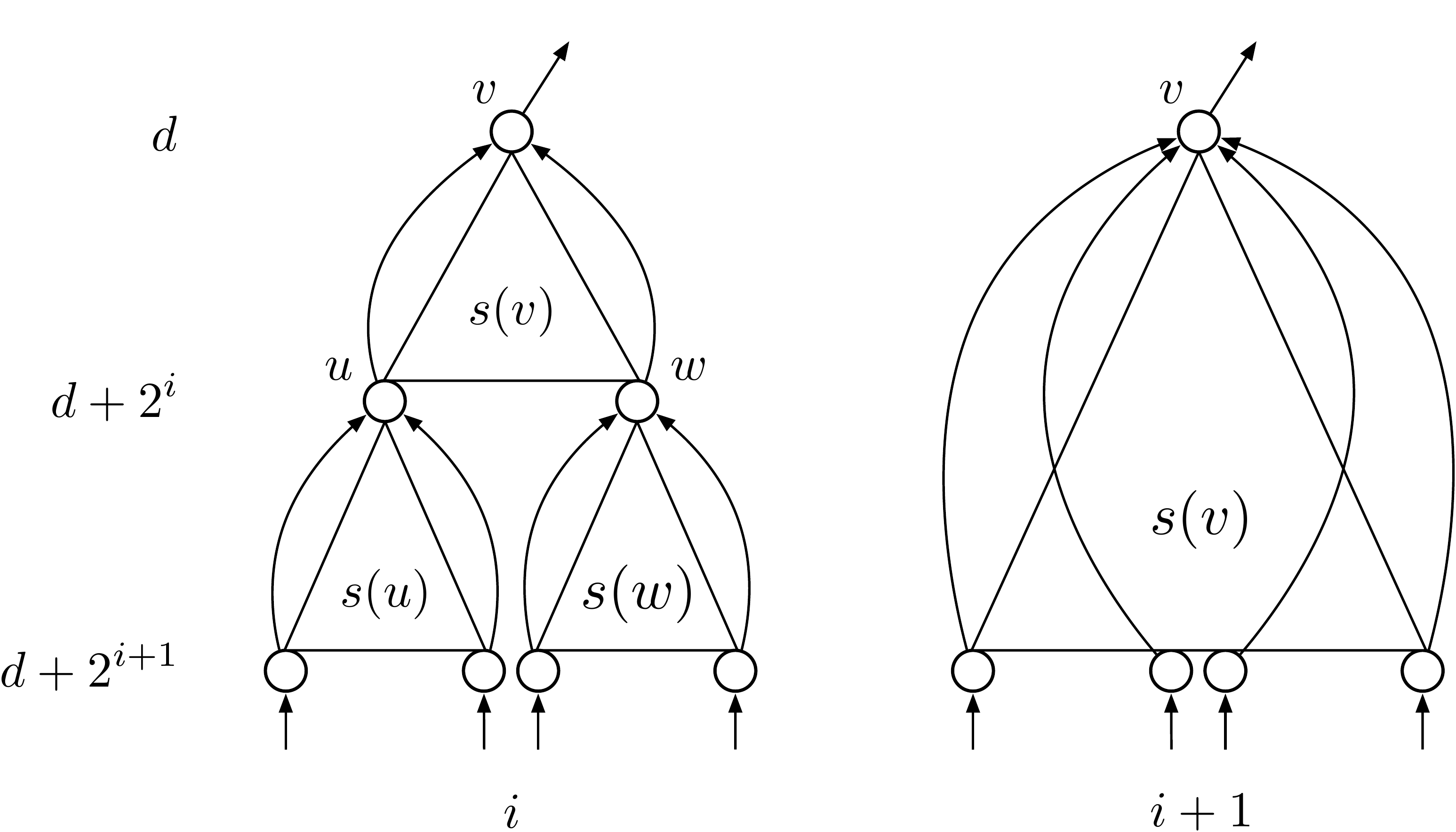}
	\caption{An illustration of an update in Step 1(a) of \css. The intuition is that node $v$ learns the size of the partial subtrees hanging from every child in $C(v)$. The values $d$ and $d + 2^i$ on the left refer to the depth of nodes $v$ and $u,w$, respectively, with regards to the whole tree. The oriented edges are not the actual edges of $G$, but rather a representation of sets $C(\cdot)$. The incoming edges of node $v$ are incident to the nodes in set $C(v)$ in both iterations $i$ and $i+1$. The value $s(v)$ in iteration $i+1$ is simply the sum of $s(v)$, $s(u)$ and $s(w)$ from iteration $i$.}
	\label{fig:countsubtreesizes}
\end{figure}

Upon termination, for every heavy node $v$ it holds that $|T(v)| > n^{\delta/2}$. Note that the ancestors of heavy nodes are also heavy; heavy nodes induce a single connected component in $G$.

\begin{lemma}[\css] \label{lem:css}
	Every node $v$ learns either the exact size of $T(v)$ or that $|T(v)| > n^{\delta/2}$. The algorithm terminates in $O(\log D)$ low-space \mpc deterministic rounds using $O(n+m)$ words of global memory. 
\end{lemma}

\begin{proof}
	We first show that every active node $v$ maintains the correct values for $s(v)$ and $C(v)$ throughout the algorithm. In iteration $i=0$, values $s(v)$ and $C(v)$ are correct by initialization. During iteration $i$ in Step 1(a), $v$ updates its values only if $v$ together with all of its descendants in $C(v)$ are active, resulting in the correct values for iteration $i+1$ by construction (see \Cref{fig:countsubtreesizes}). A node $v$ marks itself as heavy only when $s(v) > n^{\delta/2}$ or when one of its descendants is heavy. Both conditions imply that $|T(v)| > n^{\delta/2}$. If neither conditions are met and $C(v)= \emptyset$ at some point, then the value $s(v)$ for node $v$ is the exact size of $T(v)$ and $v$ marks itself light.
	
	The algorithm terminates in $O(\log D)$ iterations (with each iteration taking $O(1)$ \mpc rounds), since an active node knows $T(v)$ until depth $2^i$ in iteration $i$, and the depth of a tree is $D$. Observe that for every light node $v$, it holds that $|C(v)| \leq s(v) \leq n^{\delta/2}$. Hence, local memory is never violated, because when a node updates $C(v)$ in Step 1(a), the resulting set is of size at most $n^{\delta/2} \cdot n^{\delta/2}$ nodes. Also, storing value $s(v)$ takes only $O(\log n)$ bits. Global memory is never violated, since, by design, a node $u$ is only kept in the set $C(\cdot)$ of exactly one node.
	
	Algorithm \css can be thought of as a modified version of graph exponentiation where nodes only keep track of the furthest away descendants. For the communication in Step 1(a) to be feasible, the set $C(v)$ is simply a set of IDs corresponding to the desired nodes. Observe that the communication in Step 1(a) is always initialized by $v$, and not by the descendants in $C(v)$ (nodes in $C(v)$ don't even know the ID of $v$). This is feasible, because, by design, every node has at most one ancestor that initializes communication.
\end{proof}

\subsection{Solving \lcls: \gs} \label{sec:gs}

After executing \css, by Lemma \ref{lem:css}, every node knows if it is heavy or light. Moreover, every heavy node $v$ knows if it has a light child or not. If so, node $v$ remarks itself from heavy to \textit{local root}. If there are no local roots in $G$, mark the actual root of the tree as a local root. During algorithm \gs, heavy nodes do nothing, and all other nodes (including local roots) maintain the following variables:
\begin{itemize}
	\item $i$: iteration counter
	\item $C(v)$: a subset of descendant nodes.
\end{itemize}

The procedure is as follows (see \Cref{fig:gathersubtrees} for an example).
\falgo{$\gs(G)$}{
	\item[] \hspace{-6mm} Each node $v$ initializes: $C(v) \larr$ set of light children of $v$ in $G$, and $i \larr 0$.
	\item Repeat the following steps until $C(v)$ for every local root $v$ consist of the union of subtrees $T(u)$ for every light child $u$.
	\begin{enumerate}[leftmargin=*]
		\item Every local root $v$ updates $C(v) \larr C(v) \cup \bigcup_{u \in C(v)} C(u)$.
		\item Every light node $w$ updates $C(w) \larr \bigcup_{u \in C(w)} C(u)$.
		\item Update $i \larr i + 1$.
	\end{enumerate}
}

\begin{figure}
	\centering
	\includegraphics[width=0.8\textwidth]{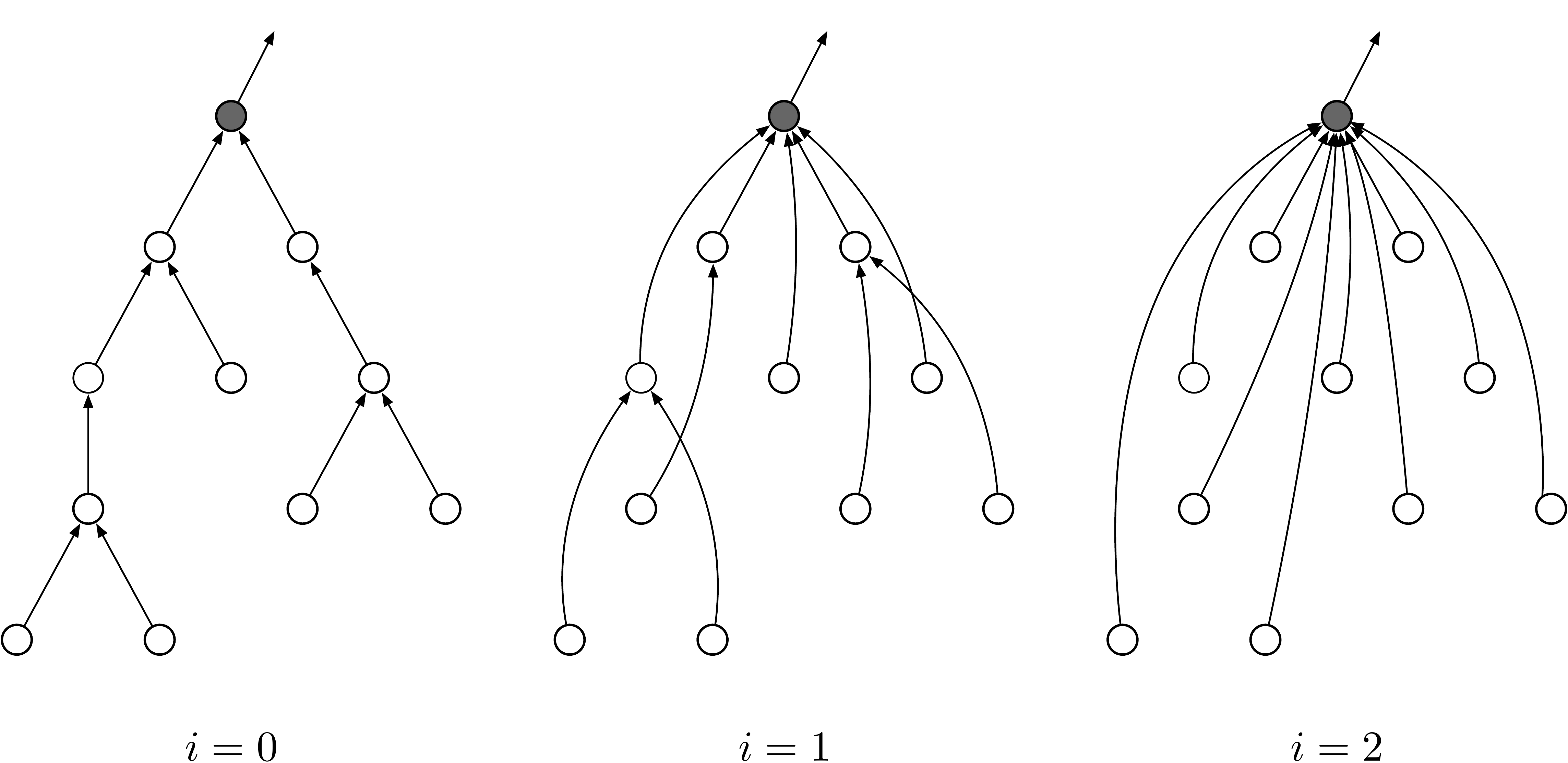}
	\caption{Three iterations of \gs, with a local root marked black. The directed edges are not necessarily the actual edges of $G$, but rather a representation of sets $C(\cdot)$. The incoming edges of a node $v$ are incident to the nodes in set $C(v)$. The figure illustrates how local roots behave differently than other nodes: local roots aggregate all descendants, while other nodes replace current ones with new ones.} 
	\label{fig:gathersubtrees}
\end{figure}

We phrase the algorithm in terms of the subtrees of the light children of a local root $v$, instead of the subtree of $v$ directly, and we do this for a simple reason: a local root $v$ may have children that are also local roots, in which case, $v$ does not want to learn anything in their direction.

\begin{lemma}[\gs] \label{lem:gs}
	Every local root has gathered the subtree $T(u)$ for every light child $u$. The algorithm terminates in $O(\log D)$ low-space \mpc deterministic rounds using $O(n + m)$ words of global memory.
\end{lemma}

\begin{proof}
	By design, the set $C(v)$ of a local root $v$ contains $T(u)$ until depth $2^i - 1$ for every light child $u$ in iteration $i$. Since the depth of a tree is $D$, after at most $O(\log D)$ iterations, for every local root $v$, $C(v)$ contains $T(u)$ for every light child $u$. 
	
	Observe that $|T(u)| \leq n^{\delta/2}$ for every light node $u$. This implies that $|C(v)| \leq n^\delta$ for every local root $v$, since it has a constant number of children (the maximum degree of the graph is constant). Hence, local memory is respected. Global memory is never violated, since by design, a node $u$ is only kept in set $C(\cdot)$ of exactly one node.
	
	Similarly to \css, algorithm \gs can be thought of as a modified version of graph exponentiation. However, as opposed to \css, algorithm \gs actually gathers the whole subtrees into the memory of preselected nodes (local roots). Similarly to \css, we store IDs in the sets $C(\cdot)$ in order for the communication in Step 1 to be feasible. Also, the communication is made possible due to every node having at most one ancestor that initializes the communication.
\end{proof}

\subsection{Solving \lcls: \cs} \label{sec:cs} 

After executing \gs, by Lemma \ref{lem:gs}, every local root has gathered the IDs of the nodes in the subtree $T(u)$ for every light child $u$. 

\falgo{$\cs(G)$}{
	\item Perform the following step for every light child $u$ of every local root $v$.  Denote $e^*=(u,v)$. Every local root $v$ gathers the topology of $T(u)$, along with $\phi(w)$ for every node $w$ in $T(u)$ and $\psi(e)$ for every edge $e$ in $T(u)$. Every local root $v$ computes the set of labels $L(e^*)$ satisfying that, by labeling the half-edge $(v,e^*)$ with a label in $L(e^*)$, it is possible to complete the labeling in $T(u)$ by only using configurations allowed by $\phi$ and $\psi$.
	\item Every local root $v$ updates $\phi(v)$ by possibly discarding some tuples. Let $P(v)$ be the set of ports of $v$ connecting it to light nodes, and let $e_i$ be the edge reached from $v$ by following port $i$. A tuple $(\ell_1,\ldots,\ell_d)$ is kept in $\phi(v)$ if and only if $\ell_i \in L(e_i)$ for all $i\in P(v)$.
}

\begin{lemma} \label{lem:cs-compat}
    Let $\phi$ and $\psi$ be the compatibility tree before performing \cs, and let $\phi'$ be the updated compatibility tree after performing \cs. 
	Let $G'$ be the resulting graph after performing \cs. The following holds.
	\begin{itemize}
	    \item If there exists a correct solution for the input graph $G$ according to $\phi$ and $\psi$, then there exists also a correct solution for $G'$ according to $\phi'$ and $\psi$.
	    
	    \item Given any correct solution for $G'$ according to $\phi'$ and $\psi$, it can be transformed into a correct solution for $G$ according to $\phi$ and $\psi$.
	\end{itemize} 
\end{lemma}

\begin{proof}
In order to show the first statement, suppose that there is a correct solution for $G$. Notice that, since $G'$ is a subgraph of $G$, if, for every $v\in G'$, we use the tuple in $\phi(v)$ of the correct solution in $G$, then we get a correct solution for $G'$. Hence, we need to ensure that the tuple used by $v$ is in $\phi'(v)$. But this is exactly what we do: every tuple excluded from $\phi(v)$ in Step 2 is not part of any correct solution for $G$, and hence the first statement holds. For the second statement, observe that from the definition of $G'$, it follows that any correct solution for $G'$ provides a partial solution for $G$ (all labels are fixed except the ones in the compressed subtrees), and this partial solution is part of a correct solution for $G$. Hence, a correct solution for $G$ can be obtained by extending the provided solution to the compressed subtrees. An extension is guaranteed to exist, since all non-extendable tuples of $\phi(v)$ were removed previously in Step 2. Note that this extension can be performed by all local roots simultaneously, since there are no dependencies between subtrees.  
\end{proof}

\begin{lemma}[\cs] \label{lem:cs}
	The algorithm terminates in $O(1)$ time in the low-space \mpc model and uses $O(n+m)$ words of global memory.
\end{lemma}

\begin{proof}
	Gathering the topology of $T(u)$, along with $\phi(w)$ for every node $w$ in $T(u)$ and $\psi(e)$ for every edge $e$ in $T(u)$ is possible: by \Cref{lem:gs}, the local root $v$ knows the IDs of all nodes in $T(u)$, and hence node $v$ can simply gather all incident edges from all nodes in $T(u)$ in constant time, and reconstruct $T(u)$ locally. This does not break any memory constraints, since $v$ receives every edge from at most two nodes, the number of edges is bounded by the number of nodes, and the sets $\phi(w)$ and $\psi(e)$ are of constant size. Computing the sets $L(\cdot)$ does not require communication, and can be done locally in constant time. Finally, the removal (contraction) of nodes and updating the set $\phi(v)$ also takes constant time, concluding the runtime proof. Moreover, the global memory is not violated, since similarly to \gs, a node $u$ is gathered by only one local root $v$.
\end{proof}

\subsection{Solving \lcls: \acp} \label{sec:cp}

The aim of this algorithm is to compress all paths into single edges while retaining the compatibility information of the paths, i.e., if the  problem is solved, the solution can be extended to the paths that were compressed. As opposed to the previous subroutines, \acp does not capitalize on anything that is done by the previous routines. 

We begin with a slight detour and first show how all degree-2 nodes can compute their distance to the highest ID endpoint with the following algorithm. To keep things simple, we present an algorithm for a single path $H$ with two endpoints of degree 1.

\falgo{$\cd(G)$}{	
	\item Define degree-2 nodes as \textit{internal} nodes, degree-1 nodes as \textit{endpoints}, and the higher ID endpoint as the \textit{head}. The head is denoted by $h$. Assign weight $w(e) \larr 1$ for every edge $e$. Repeat the following steps until all internal nodes share a weighted edge with both endpoints; the weight of the edge equals the distance. 
	\begin{enumerate}[leftmargin=*]
		\item Every internal node with incident edges $e = \{u,v\}$ and $e' = \{v,w\}$, removes$^*$ $e$ and $e'$ from $H$ and replaces them with a new edge $e'' = \{u,w\}$ and sets $w(e'') \larr w(e) + w(e')$. \item[] $^*$If $u$ (or $w$) is an endpoint, node $v$ does not remove $e$ (or $e'$).
	\end{enumerate}
}

\begin{lemma}[\cd] \label{lemma:countdistances}
	Every degree-2 node knows its distance to both endpoints. The algorithm  does not require prior knowledge of $D$, terminates in $O(\log D)$ low-space \mpc rounds using $O(n+m)$ words of global memory.
\end{lemma}

\begin{proof} 
	The weight of each edge $\{v,u\}$ in the graph equals the number of edges between $v$ and $u$ in the original graph. The base case being evident from the initialization of the weight of each edge, and the induction step from the update step $w(e'') \larr w(e) + w(e')$. Since the shortest path from either endpoint to any other node in the path decreases by a factor of at least $3/2$, \cd terminates in $O(\log D)$ time. 
	
	Creating edges and communicating through them can be done in constant time in \mpc, since storing an edge equals storing the ID of the neighbor. Observe that every internal node keeps exactly two edges in memory. In order to not break their local memory, endpoints do not keep track of any edges. Since all nodes keep at most two edges in memory, global memory is respected.
\end{proof}

Next, we show that, by using the distances computed with \cd, we can compress paths of all lengths in $O(\log D)$ time, while respecting both local and global memory. Recall that for a path, $h$ denotes its highest ID endpoint.

\falgo{$\acp(G)$}{
	\item Set $i \larr 0$ and execute the following steps for every path $H$ until it consists of one edge. 
	\begin{enumerate}
		\item Define an MIS set $Z_i \coloneqq \{ v \in H \mid \deg(v) = 2 \text{ and } \text{$d_G(v,h)$ is \textit{not} divisible by $2^{i+1}$} \}$.
		\item Every node $v \in Z_i$ with incident edges $e = \{u,v\}$ and $e' = \{v,w\}$ removes $e$ and $e'$ from $G'$ and replaces them with a new edge $e'' = \{u,w\}$ (with port $1$ connected to $u$ and port $2$ connected to $w$). Furthermore, for the new edge, $v$ sets $\psi(e'')$ to the set of all tuples $(\ell_1,\ell_2)$ satisfying that there exist two labels $x$, $y$ satisfying the following. Let $p_1$ be the port connecting $e$ to node $u$, let $p_2$ be the port connecting $e$ to $v$, let $p_3$ be the port connecting $v$ to $e$, let $p_4$ be the port connecting $v$ to $e'$, let $p_5$ be the port connecting $e'$ to $v$, and let $p_6$ be the port connecting $e'$ to $w$:
		\begin{itemize}
		    \item There is a tuple in $\psi(e)$ with label $\ell_1$ in position $p_1$ and label $x$ in position $p_2$;
		    \item There is a tuple in $\phi(v)$ with label $x$ in position $p_3$ and label $y$ in position $p_4$;
		    \item There is a tuple in $\psi(e')$ with label $y$ in position $p_5$ and label $\ell_2$ in position $p_6$.
		\end{itemize}
		\item Update $i \larr i + 1$.
	\end{enumerate}
}

From the perspective of the nodes $u$ and $w$, the new edge $e''$ replaces the old edges $e$ and $e'$, respectively. In other words, if $u$ was connected through port $j$ to edge $e$, now it is connected through port $j$ to edge $e''$. Notice that, the reason for which, at each step, we compute an MIS, is that, if MIS nodes replace their two incident edges of the path with a single edge, we still obtain a (shorter) path as a result.

\paragraph{Consecutive MIS.} Executing \cd gives us the means to compute consecutive maximal independent sets in \acp, which is not exactly obvious nor easily attainable using other means. If we were to compute an MIS directly with, e.g., Linial's \cite{linial} algorithm in every iteration of Step 2, we would end up with a total runtime of $O(\log D \cdot \log^* n)$. An alternative approach would be to employ the component-unstable $O(1)$-time algorithm that computes an independent set of size $\Omega(n/\Delta)$ by \cite{componentstable}. This approach also fails for multiple paths, since the algorithm in \cite{componentstable} does not give the guarantee that a constant fraction of nodes in \textit{all} paths join the independent set, leading to a total runtime of $O(\log n)$. To summarize, \cd is a novel approach to a very non-trivial problem, yielding component stability and a sharp $O(\log D)$ runtime.

\begin{lemma} \label{lem:cp-compat}
    Let $\phi$ and $\psi$ be the compatibility tree before performing \acp, and let $\psi'$ be the updated compatibility tree after performing \acp. 
	Let $G'$ be the resulting graph after performing \acp. The following holds.
	\begin{itemize}
	    \item If there exists a correct solution for the input graph $G$ according to $\phi$ and $\psi$, then there also exists a correct solution for $G'$ according to $\phi$ and $\psi'$.
	    
	    \item Given any correct solution for $G'$ according to $\phi$ and $\psi'$, it can be transformed into a correct solution for $G$ according to $\phi$ and $\psi$.
	\end{itemize}
\end{lemma}

\begin{proof}
     Assuming that $Z_i$ is indeed an MIS, the first statement follows from the fact that acting nodes (i.e., MIS nodes) are never neighbors and every set $\psi(\{u,w\})$ that an acting node $v$ creates only discards configurations that do not correspond to valid solutions for the subpath $(u,v,w)$. The second statement holds by the definition of labels $\ell_1,\ell_2,x,y$, since we can perform the process in reverse.
	
	Let us show that $Z_i$ constitutes an MIS. For any $i$, observe that the distances of the remaining nodes constitute all multiples of $2^i$ up until some number (the length of the path). Hence, every second node is not divisible by $2^{i+1}$ and joins $Z_i$, proving the statement.
\end{proof}

\begin{lemma}[\acp] \label{lem:cp}
	There are no degree-2 nodes left in the graph. The algorithm terminates in $O(\log D)$ low-space \mpc rounds using $O(n+m)$ words of global memory. 
\end{lemma}

\begin{proof}
	Since $Z_i$ constitutes an MIS, every path shortens by a constant factor. After $O(\log D)$ iterations, every path is compressed into a single edge. Every iteration consists of a constant number of communication rounds, every node uses a constant amount of memory, and compressing paths into edges never creates new degree-2 nodes.
\end{proof}

\subsection{Solving \lcls: \dcp} \label{sec:dcp}

Assuming that the problem of interest is solved in the current graph, we essentially reverse \acp and iteratively extend the solution from certain edges to the paths that were previously compressed into those edges. By ``the problem is solved in the current graph'' we simply mean that the output labels of the half-edges in the current graph satisfies \Cref{def:partialsol}.

\falgo{$\dcp(\dot{G})$}{
	\item All nodes that performed \acp in phase $k$, know the last iteration $i$ and can perform the following until $i=0$.
	\begin{enumerate}
		\item Every node $v \in Z_i$ learns the fixed half-edge labels $(\ell_1,\ell_2)$ assigned to $e''=(u,w)$ ($e''$ is the edge $v$ had created). Node $v$ removes $e''$ from the graph and replaces it with $e=(u,v)$ and $e'=(v,w)$ (edges $e$ and $e'$ are the edges $v$ had removed). Furthermore, $v$ assigns half-edge labels $\ell_1,x$ to edge $e$ and labels $y,\ell_2$ to edge $e'$ such that the labeling satisfies $\psi(e)$, $\phi(v)$, and $\psi(e')$.
		\item Update $i \larr i-1$
	\end{enumerate}
}

\begin{lemma}[\dcp] \label{lem:dcp}
	The \lcl problem on the graph is solved according to \Cref{def:partialsol}, and the graph has the same node and edge sets as $G$ in phase $k$ before executing \acp. The algorithm terminates in $O(\log D)$ low-space \mpc rounds using $O(n+m)$ words of global memory. 
\end{lemma}

\begin{proof}
	All we do is reversing the steps of \acp and extending the solution for $\dot{G}$ to the decompressed paths, resulting in a correct solution on the graph that has the same node and edge sets that we had before the compression. Since the computation of the solution is done locally, and extending the solution requires a constant amount of memory and communication, the lemma follows from \Cref{lem:cp-compat,lem:cp}.
\end{proof}

\subsection{Solving \lcls: \dcs} \label{sec:dcs}

The assumption for this algorithm, similarly to \dcp, is that the  \lcl problem on the graph is solved correctly according to \Cref{def:partialsol}. We also assume, and keep the invariant, that we do not only know the partial output assignment given to $\dot{G}$, but we also know, for each node $v$ of $G$, the tuple $c(v) \in \phi(v)$ assigned to it. This is especially useful at the beginning, when we have the root that is a singleton, and hence has no incident edges in the current graph, but we still want to know how to label its incident half-edges after decompressing the subtrees rooted at its children.

\falgo{$\dcs(\dot{G})$}{
	\item Every local root $v$ of phase $k$ decompresses every subtree $T(u)$ compressed into it during phase $k$, while simultaneously solving the \lcl problem on $T(u)$.
}

\begin{lemma}[\dcs] \label{lem:dcs}
	The \lcl problem on the graph is solved according to \Cref{def:partialsol}, and it has the same node and edge sets as $G$ in phase $k$ before executing \cs.  The algorithm terminates in $O(1)$ low-space \mpc rounds using $O(n+m)$ words of global memory. 
\end{lemma}

\begin{proof}
	The first statement follows from Lemma \ref{lem:cs-compat}, since the solution for $\dot{G}$ can be extended to the decompressed trees, resulting in a correct solution on the graph that has the same node and edge sets that we had before the compression. For each node $u$ in the decompressed trees, we store in $c(u)$ the tuple used to label its incident half-edges. The runtime follows from Lemma \ref{lem:cs}, and the memory is respected trivially.
\end{proof}

\subsection{Proof of \texorpdfstring{\Cref{thm:LCLSolver}}{Lg}} \label{sec:proofsec}

By the lemmas in \cref{sec:cs,sec:gs,sec:cs,sec:cp,sec:dcp,sec:dcs}, the problem is solved in the original forest, and all subroutines of \sr have time complexity $O(\log D)$ in the low-space \mpc model and use $O(n+m)$ words of global memory. All of the subroutines are clearly deterministic. What is left to prove is that
\begin{itemize}
	\item[(i)] after a constant number of phases in Step 2, the graph is reduced to a single node;
	\item[(ii)] after a constant number of repetitions in Step 4, the graph is expanded to its original form;
	\item[(iii)] if the input graph is a forest, the algorithm is component-stable, and the runtime becomes $O(\log D_{\max})$, where $D_{\max}$ denotes the maximum diameter of any component.
\end{itemize}

\begin{proof}[Proof of (i)]
	The proof is very similar to the proof of \Cref{lem:lConstant}, but we restate the claims for completeness. Let us recall what effectively happens during a phase. There are only two subroutines that alter the graph: in \cs, all subtrees of size $\leq n^{\delta/2}$ are compressed into the first ancestor $v$ with a subtree of size $> n^{\delta/2}$; then, in \acp, all paths are compressed into single edges, leaving no degree-2 nodes in the graph. Let $G_j$ and $n_j = |G_j|$ denote the graph and the size of the graph at the beginning of phase $j$, respectively. We claim that after one phase, the number of nodes in the graph drops by a factor of $\Theta(n^{\delta/2})$. Observe, that after \cs every leaf $w$ in the graph corresponds to a subtree of size $\geq n^{\delta/2}$ that was removed. Moreover, the same holds also after \acp. Hence,
	\begin{align*}
		n_{j} &\geq n_{j+1} +  |\{w \in G_{j+1} \mid \deg_{G_{j+1}}(w)=1\}| \cdot  n^{\delta/2} \\
		&> n_{j+1} + n^{\delta/2} \cdot n_{j+1}/2 \cdot  \\
		&= n_{j+1} (1+n^{\delta/2}/2)\text{,}
	\end{align*}
	implying that
	\begin{align*}
		n_{j+1} \leq \frac{n_j}{1+n^{\delta/2}/2}~.
	\end{align*}

	The first strict inequality stems from the fact that there are no degree-2 nodes left after phase $j$, and hence the number of leaf nodes in $G_{j+1}$ is strictly larger that $n_{i+1}/2$. It is clear that after $O(1/\delta)$ phases, the graph is reduced to one node. 
\end{proof}

\begin{proof}[Proof of (ii)]
	Let us recall what effectively happens during a repetition. Both subroutines \dcp and \dcs alter the graph by decompressing the paths and subtrees that were compressed previously in some phase. Hence, the number of repetitions is equal to number of phases, which is constant.
\end{proof}

\begin{proof}[Proof of (iii)]
	During the algorithm, the only communication between the components happens in order to start the subroutines in synchrony, which does not affect the \lcl solution. It does however affect the runtime, since smaller components may be stalled behind larger components. Hence, in all runtime arguments, $D$ can be substituted with $D_{\max}$.
\end{proof}

\paragraph{Extension to unsolvable \lcl problems.}

If the \lcl problem is unsolvable, we can detect it in the following way. If, during any phase of Step 2, a local root $v$ ends up with an empty set of tuples in $\phi(v)$, the original \lcl problem must be unsolvable. Node $v$ can then broadcast to all nodes in the graph to output label $\bot$ on their incident half-edges, indicating that there is no solution to the \lcl problem. 

\section{Conditional Hardness Results} \label{sec:Hardness}
In this section, we show that our algorithm for solving all \lcl{}s is optimal, assuming a widely believed conjecture about \mpc. 
By earlier work, we consider the following more convenient problem that is also hard under the conjecture.
We note that, due to technical reasons, our problem definition is slightly different to the one in~\cite{Ghaffari2019}.
Following in the footsteps of previous work, we will show that our version of the problem is also hard under the \vs conjecture.

\begin{definition}[The \diamConnectivity problem]\label{def:hardproblem}
Consider a graph that consists of a collection of paths of diameter $O(D)$, for some parameter $D$ satisfying $D \in \Omega(\log n)$ and $D = n^{o(1)}$.
Given two special nodes $s$ and $t$ of degree $1$ in the graph, the algorithm should provide the following guarantee:
If $s$ and $t$ are in the same connected component, then the algorithm should output YES. 
If $s$ and $t$ are in different connected components, the algorithm should output NO.
\end{definition}

\begin{lemma}\label{lemma: path-hardness}
    Assuming that the \vs conjecture holds, there is no deterministic low-space \mpc algorithm with $\poly(n)$ global memory to solve the \diamConnectivity problem in $o(\log D)$ rounds.
\end{lemma}
\begin{proof} 
    On a high level, we show that we can use an algorithm for the \diamConnectivity problem to reduce the size of the given cycles by a multiplicative factor $D$, unless the given cycles are already too small. We then show that, by recursively applying this algorithm, we obtain a solution for the \vs cycles problem.
    
    In more detail, we are given a graph $G$ that is either one or two cycles, where each cycle is of length at least $n/2$. Let $D$ be in $\Omega(\log n)$ and in $n^{o(1)}$. We proceed in phases, starting in phase $i=0$. We assume that at the beginning of phase $i$ the graph $G$ contains at least $(n / D^i) / 2^i$ nodes and at most  $(n / D^i) \cdot 2^i$ nodes, and we guarantee that at the end of phase $i$ the graph contains at least $(n / D^{i+1}) / 2^{i+1}$ nodes and at most  $(n / D^{i+1} )\cdot 2^{i+1}$ nodes. If the graph contains two cycles, this factor-$D$ reduction will actually independently hold for the size of each cycle. This is performed by running the algorithm $A$ that solves the \diamConnectivity problem. We stop when the number of nodes is $n^{o(1)}$, which requires $O(\log_D n)$ phases. Then, we can spend $o(\log n)$ rounds to solve the problem with known techniques (e.g.,~\cite{Behnezhad2019}). If the \diamConnectivity problem could be solved in $o(\log D)$ rounds, we would obtain a total running time of $O(\log_D n) \cdot o(\log D) + o(\log n) = o(\log n)$, which violates the \vs conjecture. We now explain a single phase of the algorithm.

    In each phase $i$, we maintain the invariant that, if there are two cycles, the larger one contains at most $4^i$ times the nodes of the smaller one.
    Assume that at the beginning of phase $i$ there are at least $(4^{i} + 1)c D \log n$ nodes, for a sufficiently large constant $c$. If it is not the case, then we are done, because the number of nodes is in $n^{o(1)}$.
    
    Sample the nodes in $G$ with probability $1/D$ and turn each sampled node \emph{inactive}. At the end of the phase, only inactive nodes will remain, and by a standard Chernoff bound, with high probability, the number of inactive nodes is at least a factor $D/2$ and at most a factor $2D$ smaller than the original amount of nodes. Moreover, this holds independently on each cycle, and hence the ratio of the sizes of the obtained cycles can increase by at most a factor $4$, hence maintaining the invariant.
    
    We now show an upper bound on the length of the obtained paths, induced by active nodes. Consider a sequence of $c' D \log n$ nodes, for some sufficiently large constant $c'$. The probability that none of them is sampled is $(1 - 1/D)^{c' D \log n}$, and hence, with high probability, each path has length $O(D \log n)$.
    Moreover, since the shortest cycle has at least $c D \log n$ nodes, then, by fixing $c$ sufficiently larger than $c'$, we obtain that each cycle contains at least one sampled node with high probability, and hence $G'$ is a collection of paths, as required.

    We create many instances of the \diamConnectivity problem from these paths as follows.
    Fix a node $u$ with degree $1$ in $G'$.
    We set $u \coloneqq s$ and create an instance of \diamConnectivity for each possible choice of $t \neq s$, where $t$ is also a degree $1$ node.
    Notice that there can be at most $n$ of such choices.
    Furthermore, we do the same construction for all possible choices of $s$, which results in $O(n^2)$ instances of the \diamConnectivity problem.
    
    Suppose now that we have a deterministic $o(\log D)$ time algorithm $A$ to solve the \diamConnectivity problem. The paths have length $O(D \log n)$, and hence running this algorithm requires $o(\log (D \log n))$ rounds, which, by the assumption on $D$, is still in $o(\log D)$.
    Run $A$ independently on each of the $O(n^2)$ instances of the \diamConnectivity problem.
    On an instance where $s$ and $t$ are on the same path, the algorithm returns YES and otherwise NO.
    Hence, we can derive which endpoints in $G'$ are on the same path in $o(\log D)$ time.
    
    Then, we create a new instance of the \vs cycle problem as follows.
    For each pair $s$ and $t$ on the same path, we create a virtual edge between the inactive neighbors of $s$ and $t$ and remove the active nodes.
    The number of nodes decreases at least by a factor $D/2$ and at most by a factor $2D$, as required.
\end{proof}

Since a connected component algorithm clearly solves the \diamConnectivity problem, we obtain the following corollary.
\begin{corollary}
Assuming the \vs conjecture, there is no low-space memory \mpc algorithm to solve connected components in $o(\log D)$ rounds on forests.
\end{corollary}

We now show that we can define an \lcl{} problem $\Pi$ for which we can convert any solution into a solution for the problem of \Cref{def:hardproblem} in constant time. This implies a conditional lower bound of $\Omega(\log D)$ for $\Pi$,  implying also that our generic solver, that runs in $O(\log D)$ rounds, is optimal. Instead of defining $\Pi$ by defining $C_V$ and $C_E$ formally, which makes it difficult to parse the definition, we provide a human understandable description of the constraints.
\begin{itemize}
    \item The possible inputs of the nodes are $0$ or $1$. In the instances that we create, all nodes will be labeled $0$, except for $s$, which will be labeled $1$.
    \item The possible outputs are on edges, and every edge needs to be either oriented or unoriented.
    \item All nodes of degree $2$ must have either both incident edges unoriented, or both incident edges oriented. If they are oriented, one must be incoming and the other outgoing.
    \item Any node of degree $1$ with input $1$ must have its incident edge oriented.
    \item Any node of degree $1$ with input $0$ must have its incident edge either unoriented, or oriented incoming.
\end{itemize}
We can observe some properties on the possible solutions for this problem:
\begin{itemize}
    \item The edges of a path are either all oriented or all unoriented.
    \item The edges of a path containing only nodes with input $0$ must all be unoriented, because a path needs to be oriented consistently, and endpoints with input $0$ must have their edge oriented incoming.
    \item All the edges of a path containing an endpoint with input $1$ must be oriented.
\end{itemize}
We can use an algorithm for $\Pi$ to solve the problem of \Cref{def:hardproblem} as follows. By giving $0$ as input to all nodes except $s$, and $1$ to $s$, and solving $\Pi$, we obtain a solution in which only the other endpoint of the path containing $s$ has an oriented incident edge, and we can hence check if this node is $t$. 
Since $\Pi$ is an \lcl, we obtain the following.
\begin{theorem}
    Assuming the \vs conjecture, there is no low-space memory \mpc algorithm to solve any solvable \lcl in $o(\log D)$ rounds on forests.
\end{theorem}

\appendix 
	
\section{\mpc Implementation Details} \label{sec:MPCdetails}

Initially, before executing any algorithm, the input graph of $n$ nodes and $m$ edges is distributed among the machines arbitrarily. By applying \Cref{def:aggTreeStructure}, we can organize the input such that every node and it's edges are hosted on a single machine, or, in the case of high degree, on multiple consecutive machines. 

\begin{definition}[Aggregation Tree Structure, \cite{MPCbasictools}] \label{def:aggTreeStructure}
		Assume that an \mpc algorithm receives a collection of sets $A_1,\dots, A_k$ with elements from a totally ordered domain as input. In an aggregation tree structure for $A_1,\dots,A_k$, the elements of $A_1,\dots,A_k$ are stored in lexicographically sorted order (they are primarily sorted by the number $i \in \{1,\dots,k\}$ and within each set $A_i$ they are sorted increasingly). For each $i \in \{1,\dots,k\}$ such that the elements of $A_i$ appear on at least 2 different machines, there is a tree of constant depth containing the machines that store elements of $A_i$ as leafs and where each inner node of the tree has at most $n^{\delta/2}$ children. The tree is structured such that it can be used as a search tree for the elements in $A_i$ (i.e., such that an in-order traversal of the tree visits the leaves in sorted order). Each inner node of these trees is handled by a separate additional machine. In addition, there is a constant-depth aggregation tree of degree at most $n^{\delta/2}$ connecting all the machines that store elements of $A_1 ,\dots, A_k$.
\end{definition}

This section is dedicated to showing how \cc can be implemented in the low-space \mpc model. We only cover routine \clst, since the implementation details for \cp, \dclst, and \dcp are simple, and included in the corresponding proofs. 

In the proof of \clst, we have reasoned that the local memory of a node never exceeds $O(n^\delta)$, and that the total memory never exceeds $O(n \cdot \ad^3)$. However, we have to also ensure that the low-space \mpc's communication bandwidth of $O(n^\delta)$ is respected throughout the routines (\Cref{lem:MPCdetails}). Also, we have to address the possibility of a node having degree $>n^\delta$, since we work with arbitrary degree trees. 

If, during some iteration of \clst, the degree of a node is $>n^\delta$, it is clearly heavy, and does not partake in the ongoing iteration. In fact, if the degree is $>n^{\delta/8}+1$, it is also heavy and does not partake. 
Hence, in the following lemma, we can assume that every node $v$ and its edges are hosted on a single machine, and that $\deg(v) \leq n^{\delta/8}+1$.

\begin{lemma} \label{lem:MPCdetails}
	The following routines can be performed in $O(1)$ low-space \mpc rounds:
	\begin{enumerate}
		\item A node can detect whether it is happy or full, 
		\item $\expo(X), X \subseteq N(v)$, 
		\item If node $v$ is added in $S_w$ for some $w$, $v$ is able to add $w$ to $S_v$.
	\end{enumerate}
\end{lemma}

\begin{proof} We prove the three statements separately. All three statements have the a common technical difficulty: it is possible for node $v$ to be included in $S_w$ for some $w$, such that $w \not\in S_v$, which causes communication bandwidth congestion. We address this common issue shortly, after reasoning about the separate challenges of each routine.
	\begin{enumerate}
		\item Since the property of being full depends solely on the size of $S_v$, it can be computed locally. In order for a node $v$ to detect if it is happy, $v$ only has to ask all nodes $w \in S_v$ for their degrees.
		\item In order for a node $v$ to perform $\expo(X), X\subseteq N(v)$, $v$ must ask a subset of nodes $w \in S_v$ for their $S_{w \narr r_w(v)}$, which is straightforward to implement.
		\item When needed, a node $w$ can inform nodes $v \in S_v$ that they have been added to $S_w$. After which it is straightforward for $v$ to add $w$ to $S_v$.
	\end{enumerate}
	
	In all of the routines above, it is possible for $v \in S_w$ for some $w$, such that $w \not\in S_v$. This can happen when $v$ does not maintain a symmetric view towards a direction in \fd in Step 1(c). This can cause $>n^\delta$ nodes querying node $v$, breaking the communication constraint of the low-space \mpc model. The following scheme resolves the issue. Let us first restate a tree structure that is useful to carry out computations on a set or on a collection of sets, in $O(1)$ low-space \mpc rounds and with $O(n+m)$ global memory. 
	
	Denote the collection of machines we are using for the algorithm as $M = \{M_1,M_2,\dots,M_l\}$. Let us allocate a new collection of empty machines $M' = \{M'_1,M'_2,\dots,M'_l\}$. For every node $w \in S_v$ of a node $v$ hosted by $M_j$, send a directed edge $(v,w)$ to $M'_j$. Let us call all edges of form $(x,y)$ as the \emph{outgoing} edges of $x$ and \emph{incoming} edges of $y$. Along with the edge, send the address of machine $M_j$.
	
	Define sets $A_1,\dots,A_k$ such that set $A_i$ contains all incoming edges of node $i$. Apply \Cref{def:aggTreeStructure} such that sets $A_1,\dots,A_k$ are stored in $M'$ in lexicographically sorted order (by the ID of $i \in \{1,\dots,k\}$ and within each set $A_i$, the elements are sorted increasingly). By \Cref{def:aggTreeStructure}, for each $i \in \{1,\dots,k\}$ such that the elements of $A_i$ appear on at least 2 different machines, there is a tree of constant depth containing the machines that store elements of $A_i$ as leafs and where each inner node of the tree has at most $n^{\delta/2}$ children. Let us denote this kind of tree as $\mathcal{A}_i$. 
	
	Observe that every $\mathcal{A}_i$ corresponds to a node $i$ that has $>n^{\delta}$ incoming edges, which is exactly the problematic case we have set out to deal with. The root of $\mathcal{A}_i$ can ask for $S_i$ from the machine in $M$ hosting node $i$, and distributes $S_i$ to all leaf nodes hosting the incoming edges (this requires a communication bandwidth of $O(n^{3\delta/2})$). We also establish a mapping from the machines in $M$ to machines in $M'$ such that the machine in $M$ hosting $u$ and $S_u$ knows the machines in $M'$ hosting edges $(u,v)$ for every $v \in S_u$. This is straightforward to implement, since when we distributed the edges to $M'$, we also distributed the corresponding addresses of machines in $M$.
	
	Let us describe what effectively happens when a node $u$ asks for $S_{v \narr r_v(u)}$ of node $v \in S_u$ if $v$ has $>n^\delta$ incoming edges. The machine $M_u \in M$ hosting node $u$ queries the machine $M_{(u,v)}' \in M'$ for $S_{v \narr r_v(u)}$, where $M_{(u,v)}'$ is the machine hosting edge $(u,v)$. Due to the design of the aggregation tree, machine $M_{(u,v)}'$ is a leaf of tree $\mathcal{A}_v$ and has at most $n^{\delta}$ elements (edges) stored on it. The queries to machine $M_{(u,v)}'$ comprise an incoming message size of $O(n^{\delta})$. Answering the queries would require a communication bandwidth of $O(n^{2\delta})$. 
	
	We can reduce the communication bandwidth of $O(n^{3\delta/2})$ (the distribution of $S_i$) and $O(n^{2\delta})$ (the leaves of $\mathcal{A}_v$ answering queries) to $O(n^\delta)$ by using $\delta/2$ instead of $\delta$ for the whole algorithm.
\end{proof}

\bibliographystyle{alphaurl}
\bibliography{mpc-cc-arxiv}
\appendix

\end{document}